\documentclass[12pt, draftclsnofoot,onecolumn]{IEEEtran}

\usepackage{graphics,colortbl,multicol,lipsum,xcolor,graphicx,color}
\usepackage[cmex10]{amsmath}
\usepackage[colorlinks=true,bookmarks=false,citecolor=blue,urlcolor=blue]{hyperref}
\usepackage{amsthm,mathrsfs,mathtools,amsbsy,amsfonts,amssymb}
\usepackage{setspace,float,pstool,lettrine,newclude}
\usepackage[normalem]{ulem}
\usepackage{latexsym,algpseudocode,gensymb,balance,multirow,bm,wasysym}
\usepackage{cite}
\usepackage[linesnumbered,ruled,vlined]{algorithm2e}
\usepackage{subcaption}
\SetKwInput{KwInput}{Input}  
\SetKwInput{KwOutput}{Output}
\SetKw{KwBy}{by}
\makeatletter
\newcommand{\nosemic}{\renewcommand{\@endalgocfline}{\relax}}% Drop semi-colon ;
\newcommand{\dosemic}{\renewcommand{\@endalgocfline}{\algocf@endline}}% Reinstate semi-colon ;
\let\oldnl\nl% Store \nl in \oldnl
\newcommand{\nonl}{\renewcommand{\nl}{\let\nl\oldnl}}% Remove line number for one line
\makeatother

\DeclarePairedDelimiterX\MeijerM[3]{\lparen\!}{\rparen}%
{\,#3\delimsize\vert\begin{smallmatrix}#1 \\ #2\end{smallmatrix}}

\newcommand\MeijerG[8][]{%
  G^{\,#2,#3}_{#4,#5}\MeijerM[#1]{#6}{#7}{#8}}

\WithSuffix\newcommand\MeijerG*[7]{%  
  G^{\,#1,#2}_{#3,#4}\MeijerM*{#5}{#6}{#7}}
\allowdisplaybreaks

\newtheorem{lemma}{Lemma}
\newtheorem{proposition}{Proposition}

\graphicspath{{figures/}}
\allowdisplaybreaks

\newcommand{\RNum}[1]{\uppercase\expandafter{\romannumeral #1\relax}}

\usepackage{mathtools}

\usepackage{accents}
\newlength{\dhatheight}
\newlength{\dtildeheight}
\newcommand{\doublehat}[1]{%
    \settoheight{\dhatheight}{\ensuremath{\hat{#1}}}%
    \addtolength{\dhatheight}{-0.2ex}%
    \hat{\vphantom{\rule{1pt}{\dhatheight}}%
    \smash{\hat{#1}}}}

\newcommand{\mathbbm}[1]{\text{\usefont{U}{bbm}{m}{n}#1}}

\IEEEoverridecommandlockouts
\begin{document}

\title{Semantic Filtering and Source Coding in Distributed Wireless Monitoring Systems}

\author{Pouya Agheli, 
\IEEEmembership{Student Member, IEEE}, Nikolaos Pappas, 
\IEEEmembership{Senior Member, IEEE}, and Marios Kountouris, \IEEEmembership{Fellow, IEEE.}
\thanks{P. Agheli and M. Kountouris are with the Communication Systems Department, EURECOM, Sophia-Antipolis, France, email: \{\texttt{pouya.agheli, marios.kountouris\}@eurecom.fr}. N. Pappas is with the Department of Computer and Information Science, Linköping University, Sweden, email: \texttt{nikolaos.pappas@liu.se}. The work of P. Agheli and M. Kountouris has received funding from the European Research Council (ERC) under the European Union’s Horizon 2020 research and innovation programme (Grant agreement No. 101003431). The work of N. Pappas is supported by the VR, ELLIIT, and the European Union (ETHER, 101096526). Part of this work is presented in \cite{agheli2022semantics,agheli2022semanticsGC}.}}

\maketitle

\begin{abstract}
The problem of goal-oriented semantic filtering and timely source coding in multiuser communication systems is considered here. We study a distributed monitoring system in which multiple information sources, each observing a physical process, provide status update packets to multiple monitors having heterogeneous goals. Two semantic filtering schemes are first proposed as a means to admit or drop arrival packets based on their goal-dependent importance, which is a function of the intrinsic and extrinsic attributes of information and the probability of occurrence of each realization. Admitted packets at each sensor are then encoded and transmitted over block-fading wireless channels so that served monitors can timely fulfill their goals. A truncated error control scheme is derived, which allows transmitters to drop or retransmit undelivered packets based on their significance. Then, we formulate the timely source encoding optimization problem and analytically derive the optimal codeword lengths assigned to the admitted packets which maximize a weighted sum of semantic utility functions for all pairs of communicating sensors and monitors. Our analytical and numerical results provide the optimal design parameters for different arrival rates and highlight the improvement in timely status update delivery using the proposed semantic filtering, source coding, and error control schemes.
\end{abstract}

\IEEEpeerreviewmaketitle

\section{Introduction}
\lettrine{G}{oal-oriented} semantic communication has recently attracted considerable attention and constitutes an information handling paradigm that has the potential to render various network processes more efficient and effective through a parsimonious usage of communication and computation resources. 
The design and the evolution of communication systems to date have been mainly driven by a maximalist approach, which sets audacious yet often hard to achieve goals and comes with inflated requirements in terms of resources, network over-provisioning, and ineffective scalability. In sharp contrast, goal-oriented semantics communication could be seen as a minimalist design approach (``less is more"), advocating a paradigm shift from extreme performance to sustainable performance, where the effective performance is maximized while significantly improving network resource usage, energy consumption, and computational efficiency. This vision has a long history dating back to Weaver's introduction of the Shannon model of communication \cite{ShannonWeaver49}. Various attempts, from different angles and using diverse tools, have been made in the past towards a semantic theory of communication \cite{Sem1,Sem2,Sem3,Sem4,Sem5,Sem6,Sem7,Sem9}; the vast majority of these endeavors remained at a conceptual level and did not lead to an elegant and/or insightful theory with immediate practical applications. Nonetheless, the quest for such theory has recently gained new impetus \cite{Popovski2020semantic,kountouris2021semantics,Qin22arxiv}, fueled by the emergence of connected intelligence systems, real-time cyber-physical systems, and interactive, autonomous multi-agent systems. 

An indispensable element to unlock the potential of goal-oriented communication is a concise, operational, and universal definition of the \emph{semantics of information} (SoI), i.e., a set of measures for the significance and usefulness of messages with respect to the goal of data exchange. Going beyond surrogate metrics, such as age of information (AoI) \cite{NowAoI,yates2021age, pappas2023age}, age of incorrect information (AoII) \cite{maatouk2022age}, quality of information (QoI) \cite{QoI}, and value of information (VoI) \cite{VoI_USSR,VoI, kosta2020cost, ayan2019age}, SoI can be leveraged so that the communication process, together with key associated functionalities (e.g., sensing, learning, processing), are adapted to the end-user goals/requirements. Empowered by networked intelligence, in semantic communication, only valuable and relevant content with respect to a goal is acquired, transported, and reconstructed, leading to a drastic reduction in the number of unnecessary bits processed and sent.
Otherwise stated, following the mantra that not all bits are equal, semantic communication may boost the ``information efficiency" of future communication systems, meaning that it could maximize the number of bits of useful information extracted and delivered per resource consumed. This new communication paradigm has the potential to redefine importance, timing, and effectiveness in future networked intelligent systems.

In this paper, we study a distributed monitoring system (DMS), in which multiple remote monitors receive status updates from multiple smart devices (e.g., sensors), each observing an information source. The updates generated by an information source may correspond to observations or measurements of a random physical phenomenon (event) and are taken from a known discrete distribution with finite support. Each status update is assigned a value reflecting its importance based on its intrinsic features, such as probability of occurrence, and on its extrinsic attributes, related to the sensor (source) at which it is generated. 
Semantic filtering is first performed at the transmitter, as a means to admit or drop only the most relevant or useful packets according to the associated monitor's application-dependent goal.
Admitted status updates are then encoded and sent to connected monitors over orthogonal block-fading channels. Different error control protocols are employed to harness packet transmission failures due to fading.
The objective of this paper is to design a semantic source coding scheme for a multiuser system with heterogeneous goals, considering the probability of occurrence of a realization at a sensor side, the probability that a monitor successfully receives an update from its connected sensor, and the rationale for which update packets are sent to the destination. 
Specifically, we consider that only a fraction of the source realizations is important for the different monitors. A set of realizations becomes more significant or relevant than the others for a certain application, where its elements could potentially vary for different goals or over different periods. A simple instance of this model is a scenario where one decision maker is interested in regular/standard information for monitoring purposes or typical actuation (normal mode), whereas the other monitor tracks the outliers that could potentially represent some kind of threat to the system or a possibly dangerous situation (alarm mode). In that case, only ``most" (``least") frequent source realizations are important for the first (second) monitor, treating the remaining ones as not informative or irrelevant. The SoI is captured here via a metric of timeliness for the received updates at the monitor(s), which in turn is a nonlinear function of AoI \cite{kosta2020cost}. 

\subsection{Related Work}
This work falls within the realm of timely source coding problem \cite{TSC1,TSC2,TSC3,bastopcu2020optimal,buyukates2020optimal,bastopcu2021selective,agheli2022semantics,agheli2022semanticsGC}. These works study the design of lossless source codes and block codes that minimize the average AoI in status update systems under different queuing theoretic considerations. References \cite{bastopcu2020optimal,bastopcu2021selective} consider a selective encoding mechanism at the transmitter for timely updates. The real-valued optimal codeword lengths that minimize the average age at the receiver are derived therein, whereas in \cite{buyukates2020optimal} an empty symbol is used to reset (or not) the age. In \cite{TSC1}, the authors consider a zero-wait update policy and find optimal source codes that achieve the minimum average age up to a constant gap, using Shannon codes based on a tilted version of the original symbol generating probability mass function (pmf). Semantic source coding is studied in \cite{agheli2022semantics}, where real-valued optimal codeword lengths that maximize a semantics-aware utility function and minimize a quadratic average length cost \cite{baer2006source} are determined. The semantic encoding problem is extended to a two-user system with heterogeneous, possibly conflicting or diverging, goals in \cite{agheli2022semanticsGC}. This paper extends our prior work into distributed monitoring systems (DMS) with multiple sensors at the transmitter side and multiple monitors at the destination, which have heterogeneous and dissimilar goals. We also delve into evaluating the importance of status update packets and propose semantic value assessment schemes for those packets depending on several features rather than assigning importance only based on the probability of occurrences, as studied in \cite{agheli2022semantics} and \cite{agheli2022semanticsGC}. Furthermore, an adaptive semantic filtering method and importance-aware packet error control protocols are employed to improve the performance of the system compared to fixed filtering, as in \cite{agheli2022semantics} and \cite{agheli2022semanticsGC}, and the importance-agnostic error control protocol used in \cite{agheli2022semanticsGC}. We then determine codeword lengths that maximize a weighted sum of semantics-aware utility functions for all pairs of communicating sensors and monitors, highlighting the performance gains of semantic filtering and source coding for a general model.

% This paper extends our prior work into distributed monitoring systems (DMS) with multiple sensors at the transmitter side and multiple monitors at the destination, which have heterogeneous and dissimilar goals. We propose simple semantic value assessment schemes for status update packets and an adaptive semantic filtering method. We then determine codeword lengths that maximize a weighted sum of semantics-aware utility functions for all pairs of communicating sensors and monitors, highlighting the performance gains of semantic filtering and source coding.

\subsection{Contributions}
The main contributions of this work can be summarized as follows. 
\begin{itemize}
    \item We perform semantic value assessment for the status update packets in two levels. At the source level, the importance of an arrival takes on the form of a \emph{meta-value} and captures flexibly the interdependence between different intrinsic and extrinsic attributes. At the link level, SoI is measured using a nonlinear function of AoI.
    \item We propose two semantic filtering mechanisms to control the flow of arrival packets, drop unimportant arrivals, and reserve the channel for informative packets. In the first (fixed-length) scheme, a set of unimportant arrivals is filtered using selective encoding. In the second (adaptive-length) scheme, a smaller fraction of the least important packets (considering the occurrence probability) are immediately dropped upon arrival, while the remaining ones may pass the filter depending on their relative importance for a predefined goal and its evolution. We show that the latter method performs better for any range of arrival rates at the expense of a small increase in the channel load.
    \item For packet error control, we employ truncated forms of simple and hybrid automatic repeat request (ARQ and HARQ) protocols. In contrast to conventional importance-agnostic protocols with a fixed number of retransmissions, we propose a scheme in which the maximum number of a packet's retransmissions is adapted to its importance with respect to the specific goal for which it is generated. Evidently, a packet with higher importance is assigned with more retransmission rounds in the event of multiple failures.
    \item We determine the real-valued optimal codeword length of each admitted arrival packet, based on its probability of occurrence, the observation probability related to the sensor tracking the realization, and the served monitor, as well as its meta-value. For that, we cast an optimization problem that maximizes a weighted sum of semantics-aware utility functions for all pairs of sensors and monitors, and obtain the solution both analytically and using a proposed simple algorithm.
\end{itemize}

%\textit{Organization:} Section~\ref{Section2} shows our proposed system model with status packet transmission and reception structures. Value assessment schemes of arrival and monitored packets at two levels are studied in Section~\ref{Section3}. Section~\ref{Section4} presents semantics-aware packet filtering and channel error control frameworks. Also, problem formulation and semantics-aware encoding design are proposed in Section~\ref{Section5}, and Section~\ref{Section6} depicts simulations and numerical results with their discussions. Finally, the paper is concluded in Section~\ref{Section7}.

\textit{Notations:} $\mathbb{R}$, $\mathbb{R}^+$, $\mathbb{R}_0^+$, and $\mathbb{Z}^+$ denote the set of real, positive real, non-negative real, and positive integer numbers, respectively. $\mathbb{E}[\cdot]$ is the expectation operator, $\mathcal{W}_0(.)$ denotes the principal branch of Lambert $\mathcal{W}$ function, and $\mathcal{O}(\cdot)$ denotes the growth rate of a function.

\section{System Model}\label{Section2}
We consider a DMS, in which a set $\mathcal{K}$ of smart sensor modules (SSMs), with cardinality $|\mathcal{K}|=K$, provides timely status updates to a set $\mathcal{M}$ of $M$ monitor modules (MMs) (see Fig.~\ref{fig:fig1}).
An SSM consists of a sensor, a semantic filtering module, a source encoder, and a module for information transmission (PHY - physical layer) operations, such as modulation and channel coding. An MM could be a \emph{display device} or an \emph{actuator}, which serves a specific goal and performs tasks based on status updates and commands received from the SSMs. 

We consider a model where the $k$-th SSM, $k = 1, 2, \ldots, K$, is connected to a subset of MMs, denoted by $\mathcal{M}(k)$, based on a fixed topology. Therefore, the $m$-th MM, where $m = 1, 2, \ldots, M$, receives timely status updates only from its serving SSMs, denoted as $\mathcal{K}(m) = \{k: m \in \mathcal{M}(k)\}$. An SSM tracks a physical phenomenon (event) with finite realizations and generates independent and identically distributed (i.i.d.) status update packets. The realization set for the $k$-th SSM is defined as $\mathcal{X}^{(k)} = \{x_i^{(k)} ~ |~ i\!\in\!\mathcal{I}_{n_k}\}$, $\mathcal{I}_{n_k} = \{1,2,...,n_k\}$, where each element has a probability of occurrence $\bar{p}_i^{(k)} = P_X(x_i^{(k)})$ with $P_X(\cdot)$ being a known pmf. Also, $n_k$ denotes the maximum number of realizations the $k$-th SSM can observe. Depending on the task(s) at the MM side or the goal specified by the application or the end-user, a realization $x_i^{(k)}$ at the $k$-th SSM connected to the $m$-th MM is attributed a certain importance (value) $v_i^{(k,m)}$ at its source level. The feature-based value assessment of each packet at an SSM is discussed in Section~\ref{Section3}. 
In this work, we consider a continuous-time system where packet generation follows a Poisson process with input rate $\lambda_k$ for the $k$-th SSM. Thus, the input arrival rate at one layer of an MM connected to a set of identical sensors is equal to the sum of the input rates from those SSMs.

The probability that the $m$-th MM observes an update packet from the $k$-th serving SSM among all the sensors it is connected to, i.e., $\mathcal{K}(m)$, is denoted by $p_k^{(m)}$, where $\sum_{k \in \mathcal{K}(m)} p_k^{(m)} = 1, \forall m$. Since the observation probability of an SSM from each of its connected MMs can be different, an arbitrary realization at that SSM is supposed to occur with different probabilities for those served MMs. In this regard, information processing transmission is performed in multiple layers and different communication channels, respectively. For instance, semantic filtering, source encoding, and PHY operations at the $k$-th SSM are performed in $|\mathcal{M}(k)|$ parallel layers, each of which is reserved for packets transmitted to one served MM. From the monitoring perspective, the sensors with similar physical features, probabilities of realizations $p_i^{(k,m)}$, and observation probabilities $p_k^{(m)}$ are considered \emph{identical}\footnote{In practice, sensors with identical characteristics can monitor the same phenomenon or process from different angles, locations, or time instants. In our analysis, a set of identical sensors can be treated as a unique sensor with shared properties.}, while the rest are called \emph{dissimilar} smart sensors. Like the SSMs, multiple layers are utilized at each MM to process packets arriving from its dissimilar sensors, whereas packets from identical sensors are processed in the same layers. A functional diagram for the $k$-th SSM and the $m$-th MM, $\forall m \in \mathcal{M}(k)$, is depicted in Figure~\ref{fig:fig2}. \begin{figure}[t!]
    \begin{minipage}{.42\textwidth}
    \centering
    \pstool[scale=0.38]{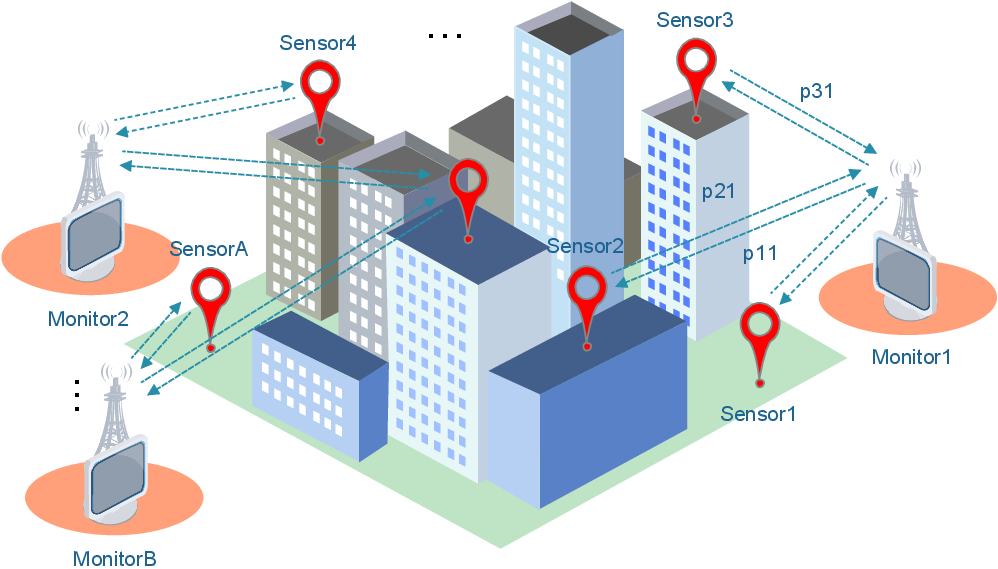}{
    \psfrag{Monitor1}{\hspace{-0.00cm}\scriptsize $\text{MM}_1$}
    \psfrag{Monitor2}{\hspace{0.0cm}\scriptsize $\text{MM}_2$}
    \psfrag{MonitorB}{\hspace{-0.00cm}\scriptsize $\text{MM}_m$}
    \psfrag{Sensor1}{\hspace{-0.04cm}\scriptsize $\text{SSM}_1$}
    \psfrag{Sensor2}{\hspace{-0.04cm}\scriptsize $\text{SSM}_2$}
    \psfrag{Sensor3}{\hspace{-0.04cm}\scriptsize $\text{SSM}_3$}
    \psfrag{Sensor4}{\hspace{-0.04cm}\scriptsize $\text{SSM}_4$}
    \psfrag{SensorA}{\hspace{-0.04cm}\scriptsize $\text{SSM}_k$}
    \psfrag{p11}{\hspace{0.08cm}\scriptsize $p_1^{(1)}$}
    \psfrag{p21}{\hspace{0.08cm}\scriptsize $p_2^{(1)}$}
    \psfrag{p31}{\hspace{0.08cm}\scriptsize $p_3^{(1)}$}
    }
    % \vspace{-0.1cm}
    \caption{A goal-oriented, semantics-empowered\\ distributed monitoring system.}
    \label{fig:fig1}
    \end{minipage}
    \begin{minipage}{.58\textwidth}
    \centering
    \pstool[scale=0.37]{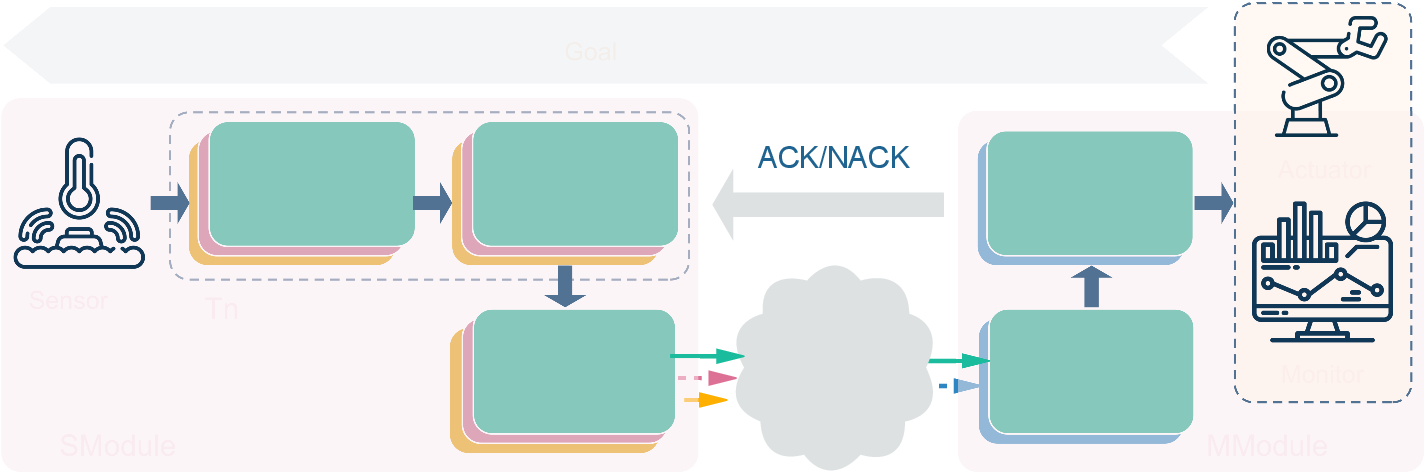}{
    \psfrag{Monitor}{\hspace{-0.18cm}\scriptsize Display}
    \psfrag{Actuator}{\hspace{-0.21cm}\scriptsize Actuator}
    \psfrag{Sensor}{\hspace{-0.06cm}\scriptsize Sensor}
    \psfrag{PHY}{\hspace{-0.09cm}\scriptsize PHY}
    \psfrag{ModNew}{\hspace{-0.07cm}\scriptsize module}
    \psfrag{Sem}{\hspace{-0.31cm}\scriptsize Semantic}
    \psfrag{filter}{\hspace{-0.07cm}\scriptsize filter}
    \psfrag{Source}{\hspace{-0.07cm}\scriptsize Source}
    \psfrag{encoder}{\hspace{-0.07cm}\scriptsize encoder}
    \psfrag{decoder}{\hspace{-0.07cm}\scriptsize decoder}
    \psfrag{P}{\hspace{-0.33cm}\scriptsize Block-}
    \psfrag{Er}{\hspace{-0.3cm}\scriptsize fading}
    \psfrag{Ch}{\hspace{-0.4cm}\scriptsize channels}
    \psfrag{ACK/NACK}{\hspace{-0.27cm}\scriptsize ACK/NACK}
    \psfrag{SModule}{\hspace{-0.2cm}\scriptsize $k$-th SSM}
    \psfrag{MModule}{\hspace{0.02cm}\scriptsize $m$-th MM}
    \psfrag{Tn}{\hspace{-0.06cm}\scriptsize Transmitter}
    \psfrag{Goal}{\hspace{-2cm}\scriptsize Application-dependent goals per layer}
    }
    % \vspace{-0.4cm}
    \caption{The functional diagram of a communication link between the $k$-th SSM and the $m$-th MM.}
    \label{fig:fig2}
    \end{minipage}
\end{figure}

\subsection{Packet Transmission}
Assuming no buffer is employed at the transmitter of every SSM, a status update packet is blocked when the channel is busy. \emph{Semantics-aware packet filtering} is employed at every SSM for each of the connected MMs as a means to transfer only the most \emph{valuable and important} packets for effectively serving the goal or purpose of the data exchange, as well as having the least possible blockage of valuable arrivals due to heavy packet load. Specifically, at the $k$-th SSM, a semantic filter admits the $l_k^{(m)}$ most important realizations via a flow controller for the $m$-th MM and discards the rest, i.e., $n_k-l_k^{(m)}$. Thus, the index set of the most valuable arrivals at the $k$-th SSM from the perspective of the $m$-th MM is denoted by $\mathcal{I}_{l_k^{(m)}}$ where $\mathcal{I}_{l_k^{(m)}}\subseteq \mathcal{I}_{n_k}$. Once semantic filtering is performed, the source encoder at the $k$-th SSM assigns instantaneous (prefix-free) codewords with lengths $\ell_i^{(k,m)}$ to the status packets $x_i^{(k)}$ admitted for the $m$-th MM, $\forall m \in \mathcal{M}(k)$, based on the following truncated distribution
\begin{align}\label{sec2:eq1}
    p_i^{(k,m)} = \begin{cases}
    p_k^{(m)}\dfrac{\bar{p}_i^{(k)}}{q_{l_k^{(m)}}}, ~\forall i \in \mathcal{I}_{l_k^{(m)}}\\
    0, ~~~~~~~~~~\forall i \notin \mathcal{I}_{l_k^{(m)}}
    \end{cases}
\end{align}
where $q_{l_k^{(m)}} \coloneqq \sum_{i\in \mathcal{I}_{l_k^{(m)}}} \bar{p}_i^{(k)}$. 

Processed packets are then forwarded to the PHY module, mapped using a binary modulation scheme, and transmitted over a \emph{noisy orthogonal} channel subject to block fading. 

\subsection{Packet Reception}\label{Section2b}
We define $t_{j-1}^{(k,m)}$ the time instant that the $j$-th packet, $j=1,2,...$, is received at the $m$-th MM from the serving $k$-th SSM. The update interval between the $j$-th successive arrival and the next one at the same layer is then modeled as a random variable (r.v.) $Y_j^{(k,m)} = t_j^{(k,m)} - t_{j-1}^{(k,m)}$. Alternatively, this interval is formed as $Y_j^{(k,m)} = W_j^{(k,m)} + S_j^{(k,m)}$, where $W_j^{(k,m)}$ and $S_j^{(k,m)}$ indicate the waiting and service time variables, respectively. The waiting time denotes the span between the $j$-th packet and the previously delivered one at the same layer, which is written as $ W_j^{(k,m)} = \sum_{a=1}^{A_{j}^{(k,m)}} Z_{a}^{(k)}$. In this definition, $A_{j}^{(k,m)}$ follows a geometric distribution with success probability $(1 \!-\! \psi_k^{(m)}) q_{l_k^{(m)}}$ and indicates the number of packets discarded until the arrival of the $j$-th packet from the selected set. Herein, $0\leq\psi_k^{(m)}<1$ denotes the semantics-aware drop factor, which is analyzed in Section~\ref{Section4a}. Furthermore, $Z_{a}^{(k)}$ is the time between two arrivals, which is exponentially distributed with rate $\lambda_k$. Therefore, the admitted arrivals are generated under a Poisson process with rate $\lambda_k (1 \!-\! \psi_k^{(m)}) q_{l_k^{(m)}}$. Besides, the service time corresponds to the duration an update packet spends in the DMS until it is completely decoded at its destination. The service time analysis and channel error control protocols are investigated in Section~\ref{Section4b}. 

\section{Semantic Value Assessment of Update Packets} \label{Section3}
In this section, we present the semantic value extraction and assessment of both arrival and successfully decoded packets. This is performed at two different scales, namely a \emph{microscopic} one, which is related to the importance of the arrivals at the information source, and a \emph{mesoscopic} one, which captures the importance of received packets at a link level \cite{kountouris2021semantics}.

\subsection{Microscopic Scale}\label{Section3a}
At the source level, the relative importance of an arrival packet depends on \emph{extrinsic} and \emph{intrinsic} features, which in turn could be either \emph{time-sensitive} or \emph{time-tolerant}. Extrinsic features are related to the smart sensor characteristics, such as spatial location, resolution and measurement quality, battery level, and reliability. Intrinsic features depend on the properties of a packet related to goal/application requirements, e.g., probability of occurrence, urgency, and loss risk. Furthermore, the importance of a packet is assumed to have a lifetime, which remains \emph{constant} until that packet is successfully received at the receiver or dropped.

\subsubsection{Meta-value assignment} 
The goal-dependent importance of each packet can take on the form of a \emph{meta-value}, which is characterized by a function of the above features. 
In general, the meta-value of each packet results from \emph{knowledge fusion} or \emph{aggregation} of $A$ intrinsic and $B$ extrinsic features. To overcome the limitations of aggregation functions in the form of weighted sums, we resort to knowledge fusion \cite{warren1999strategic}, which takes into account possible interdependence between different features/criteria and provides commensurate scales representing the attributes using the Choquet integral \cite{choquet1954theory}. Specifically, using the discrete Choquet integral, a general meta-value formula at one arrival interval for the $i$-th realization of the $k$-th SSM connected to the $m$-th MM is given by 
\begin{align}\label{eq:meta}
    v_i^{(k,m)} &= \!\underbrace{\prod_{a=1}^{A} (g_{a,i}^{(k,m)})^{{\alpha}_a^{(k,m)}}}_\text{Intrinsic features}\underbrace{\!\left(\sum_{b=1}^{B} \!\left(U_{b}^{(k,m)} \!-\!U_{b-1}^{(k,m)}\right)\! (h_{b}^{(k,m)})^{\bar{\alpha}_{b}^{(k,m)}} \!\right)\!}_\text{Extrinsic features}
\end{align}
where $\alpha_a^{(k,m)}$ and $\bar{\alpha}_b^{(k,m)}$ are constant factors. Furthermore, defining $w_b^{(k,m)}\geq0$ as the weight of the $b$-th extrinsic feature, $U_{b}^{(k,m)}$ as the weight of the $b$-th subset of extrinsic features is given by 
\begin{equation}
    U_{b}^{(k,m)} = \left(1 \!+\!\lambda_g^{(k,m)}w_b^{(k,m)}\right)\! U_{b-1}^{(k,m)} + w_b^{(k,m)}
\end{equation}
where $U_{0}^{(k,m)}=0$, $U_{1}^{(k,m)} \leq U_{2}^{(k,m)} \leq ... \leq 1$, $U_{B}^{(k,m)}=1$, and $\lambda_g^{(k,m)}\geq-1$ comes from the Sugeno fuzzy measure \cite{sugeno1974theory}, as $1 +\lambda_g^{(k,m)} = \prod_{b=1}^B\big(1 \!+\!\lambda_g^{(k,m)}w_b^{(k,m)}\big)$. For $\lambda_g^{(k,m)}=0$, we have $U_{b}^{(k,m)} = U_{b-1}^{(k,m)} + w_b^{(k,m)}$, and the extrinsic part's fusion form in \eqref{eq:meta} becomes a weighted sum.
Also, $g_{a,i}^{(k,m)}:\mathbb{R} \!\to\! \mathbb{R}_0^+$ and $h_{b}^{(k,m)}:\mathbb{R} \!\to\! \mathbb{R}_0^+$ in \eqref{eq:meta} denote value functions (VFs) of intrinsic and extrinsic features, respectively. Here, every member of $g_{a,i}^{(k,m)}$, $\forall a,i,k,m$, or $h_{b}^{(k,m)}$, $\forall b,k,m$, with known parameters is presented via a general function $\operatorname{VF}:\mathbb{R} \!\to\! \mathbb{R}_0^+$. Thus, for $A$ intrinsic and $B$ extrinsic features, we have $A+B$ different versions of $\operatorname{VF}$.
Conventionally, we model $\operatorname{VF}$ as a sum of Gaussian functions with predefined critical points $z_{n_c}$, for $n_c=1, 2, ..., N_c$, relative criticality $\varpi_{n_c}\geq0$, and minimum importance $\operatorname{VF}_{\rm min}\geq0$. For a sample point $z$, we can write
\begin{align}\label{eq:vf}
    \operatorname{VF}(z) = \dfrac{\widehat{\operatorname{VF}}(z) }{\underset{z}{\operatorname{max}}\Big\{\widehat{\operatorname{VF}}(z) \Big\}}
\end{align}
where, if $z_{\rm min} \leq z \leq z_{\rm max}$, we have
\begin{align}\label{eq:vfb}
    \widehat{\operatorname{VF}}(z) = \sum_{n_c=1}^{N_c} \varpi_{n_c} e^{-\dfrac{(z-z_{n_c})^2}{2\sigma^2}}.
\end{align}
In \eqref{eq:vfb}, the standard deviation $\sigma$ is derived such that all sample points of $\operatorname{VF}$ get higher importance than a given threshold $\operatorname{VF}_{\rm min}$. For that, with given $N_c$, $z_{n_c}$, and $\varpi_{n_c}$, $\sigma$ starts from a small value and increases with a fixed step over a finite number of iterations until $\operatorname{VF}(z) \geq \operatorname{VF}_{\rm min}$, $\forall z$. 

The proposed value assessment scheme can operate under both pull-based and push-based communication models. Assume a packet from realization $x_i^{(k)}$, $\forall i \in \mathcal{I}_{l_k^{(m)}}$ arrives at the $j$-th interval and is assigned importance value $v_j^{(k,m)}$. In the push-based model, the value of the generated packet comes from \eqref{eq:meta}, i.e., $v_j^{(k,m)}=v_i^{(k,m)}$, whereas the pull-based policy adds a constraint such that $v_j^{(k,m)}=v_i^{(k,m)}$ only if the $m$-th MM has requested an update from the $k$-th SSM at the $j$-th interval; otherwise, $v_j^{(k,m)}=0$.

\subsection{Mesoscopic Scale}\label{Section3b}
At the link level, the importance of a received packet at a destination is measured using the \emph{semantics of information} (SoI) \cite{kountouris2021semantics,pappas2021goal}. In this regard, we consider \emph{timeliness} as a contextual attribute of information, which is a non-increasing function $f_k:\mathbb{R}_0^+ \to \mathbb{R}$ of the freshness of information from the $k$-th sensor. The instantaneous SoI provided by the $k$-th SSM, $\forall k \in \mathcal{K}(m)$, for the $m$-th MM at time $t$ is modeled as $\mathcal{S}_k^{(m)}(t) = f_k(\Delta_k^{(m)}(t))$, where $\Delta_k^{(m)}(t) = t – u(t)$ denotes the instantaneous AoI at the $m$-th MM, which is defined as the difference between the current time instant $t$ and the timestamp $u(t)$ of the most recently arrived update from the considered $k$-th SSM or another identical sensor connected to that monitor.

The average SoI over an observation interval $(0, T)$, assuming a stationary ergodic process, is given by
\begin{equation}\label{sec3b:eq2}
    \bar{\mathcal{S}}_k^{(m)} = \displaystyle \underset{T\rightarrow\infty}{\lim} \dfrac{1}{T} \!\int_{0}^{T} f_k(\Delta_k^{(m)}(t)) {\rm d}t.
\end{equation}

From a monitor's perspective, the SoI at a layer reserved for the $k$-th SSM decreases according to $f_k(\cdot)$ until a valuable status update for that layer arrives. Then, the SoI rises to the value of the new update at that time. Therefore, dissimilar monitors potentially attain different SoIs over similar periods.

In this paper, we study the following forms for the SoI, namely \emph{exponential} (E-), \emph{logarithmic} (L-), and \emph{reciprocal} utility of timeliness (RUT) cases, i.e., 
\begin{align}\label{sec5:eq3}
    f_k(\Delta_k^{(m)}(t))=
    \begin{cases}
    \operatorname{exp}(-\rho_k^{(m)}(t)\Delta_k^{(m)}(t)) + \beta_k^{(m)}~~~\text{EUT case}\\
    \ln(-\rho_k^{(m)}(t)\Delta_k^{(m)}(t))+ \beta_k^{(m)}~~~~~\text{LUT case}\\
    (\rho_k^{(m)}(t)\Delta_k^{(m)}(t))^{-\kappa}+ \beta_k^{(m)}~~~~~~\text{RUT case}
    \end{cases}
\end{align}
where $\kappa \in \mathbb{Z}^+$ and $\beta_k^{(m)}$ are constant parameters, and $\rho_k^{(m)}(t)>0$ is an attenuation factor that is fixed during each update interval of the $k$-th SSM and is reinitialized for a new arrival.\footnote{The analysis is extendable for any (single or composite) non-increasing forms for the utility of timeliness.} The value of the attenuation factor at an update interval comes from the importance of that admitted arrival at the source level. Hence, we can define
\begin{align}\label{sec5:eq4}
    \!\rho_k^{(m)}(t) &\coloneqq  \rho_j^{(k,m)}\mathbbm{1}\!\left\{t_{j-1}^{(k,m)} + S_j^{(k,m)} \!\leq t \leq t_{j}^{(k,m)} + S_{j+1}^{(k,m)}\right\}
\end{align}
where $\rho_j^{(k,m)} = \rho_{\rm min} + \bigg(\underset{i \in \mathcal{I}_{l_k^{(m)}} }{\operatorname{max}}\!\!\Big\{v_i^{(k,m)}\Big\} - v_j^{(k,m)}\!\bigg)\!\big(\rho_{\rm max} \!-\! \rho_{\rm min}\big)$
for the $j$-th admitted arrival. $\rho_{\rm min}$ and $\rho_{\rm max}$ are the minimum and maximum values that the attenuation factors may attain, respectively. Evidently, the higher the importance of an arrival, the lower the attenuation factor it is assigned.

\section{Semantic Filtering and Channel Error Control}\label{Section4}
In this section, we introduce a semantics-aware packet filtering mechanism and analyze two asynchronous error control schemes.

\subsection{Semantics-Aware Filtering}\label{Section4a}
We consider two semantics-aware packet filtering schemes, which control the flow of update packet arrivals through different levels of each SSM, namely \emph{fixed-} and \emph{adaptive-length} filtering. In the light of assessed meta-values and freshness, the rationale of these filtering mechanisms is to feed the transmitter with important arrivals, while keeping the blockage rate as low as possible. The blockage rate quantifies the number of blockages occurring to the admitted packets due to heavy channel load over the total number of arrival packets.

\subsubsection{Fixed-length filtering}
A fixed-length filter merely admits packets generated from the $l_k^{(m)}$ most important realizations among all $n_k$ realizations the $k$-th SSM shares with the $m$-th MM. Packets of less important realizations, based on the given index set $\mathcal{I}_{l_k^{(m)}}$, are directly discarded, independently of their freshness (e.g., most recent realization) or the channel idleness. Therefore, even the most recent (highest freshness) yet less important packets are blocked.

\subsubsection{Adaptive-length filtering}
In the proposed adaptive-length filtering mechanism, packets' importance at both microscopic and mesoscopic scales is taken into account, which gives a chance to a subset of the least important packets, in addition to the most important ones, to be carried over the DMS. Differently from the fixed-length filtering, only a fraction of the least important packets are directly discarded in the adaptive-length filtering scheme, and a subset of the fresh yet of lower importance arrivals may adaptively pass the filter. This means that the number of admitted realizations in adaptive-length filtering is larger or equal to that of the fixed-length one. Specifically, upon the arrival of a packet generated from one of the admitted realizations when the channel is idle, the adaptive-length filter inspects whether that packet can increase the SoI at the link level after being decoded, thanks to its freshness or not. If the inspection outcome is positive, the packet passes the filter. Thus, adaptive-length filtering brings flexibility owing to which packets with lower importance but higher freshness have a chance to pass the filter.

\subsubsection{Acceptance probability analysis}
A key element of adaptive-length filtering is a parameter that quantifies the probability that a packet generated from the selected realizations is rejected (or accepted), considering previously observed updates at the link level. As mentioned in Section~\ref{Section2b}, we coin this as semantics-aware drop factor. Using this, a new arrival (status update) that does not increase the SoI upon receipt at the monitor, despite being the most recent/fresh, is rejected, and hence not transmitted.    
To make it clear, Figure~\ref{fig:fig3} illustrates the shape of packet drop by adaptive-length filtering at the $k$-th SSM connected to the $m$-th MM under three possible conditions listed below, considering a non-increasing $f_k(\cdot)$, and convex for the purpose of semantic filtering.
\begin{itemize}
    \item[I.] The first and second arrivals are of comparable importance, hence attenuation factors. Therefore, the SoI curves provided by these arrivals intersect at infinity. In this case, the filter does not discard the new arrival since it may offer a higher value upon receipt at its monitor.
    \item[II.] The $j$-th arrival has higher importance than the $(j\!+\!1)$-th one, where the primary arrival crosses the latter one before both reach $\beta_k^{(m)}$. However, the cross point is after the service time of the new arrival, i.e., $t_{j}^{(k,m)} + S_{j+1}^{(k,m)}$. Therefore, the filter keeps the $(j\!+\!1)$-th packet since it can increase $\mathcal{S}_k^{(m)}(t)$ after its successful decoding.
    \item[III.] The $j^\prime$-th arrival has higher importance compared to the $(j^\prime\!+\!1)$-th one; thus, the primary arrival crosses the latter one at a point before the delivery of that new packet. In that case, the filter discards the $(j^\prime\!+\!1)$-th packet since it cannot bring a higher value at the link level even after its successful delivery. Thanks to this mechanism, the system obtains better $\mathcal{S}_k^{(m)}(t)$, as highlighted in the figure.
\end{itemize}

\begin{figure}[t!]
    \centering
    \pstool[scale=0.8]{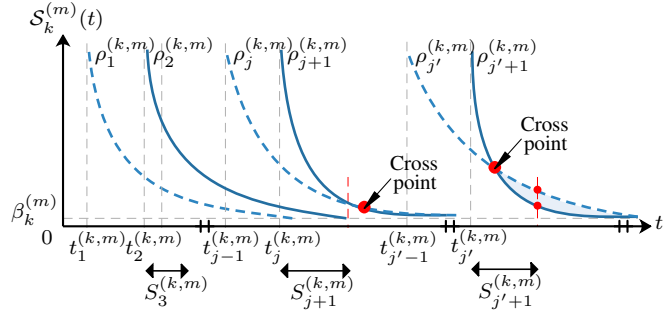}{
    \psfrag{SoI2}{\hspace{-0.2cm}\scriptsize $ \mathcal{S}_k^{(m)}(t)$}
    \psfrag{t}{\hspace{0.04cm}\scriptsize $t$}
    \psfrag{0}{\hspace{0.0cm}\scriptsize $0$}
    \psfrag{t12}{\hspace{-0.05cm}\scriptsize $t_{1}^{(k,m)}$}
    \psfrag{t22}{\hspace{-0.05cm}\scriptsize $t_{2}^{(k,m)}$}
    \psfrag{t1j-1}{\hspace{-0.05cm}\scriptsize $t_{j-1}^{(k,m)}$}
    \psfrag{t2j-1}{\hspace{-0.05cm}\scriptsize $t_{j-1}^{(k,m)}$}
     \psfrag{t2j-1p}{\hspace{-0.05cm}\scriptsize $t_{j^\prime-1}^{(k,m)}$}
    \psfrag{t1j}{\hspace{-0.05cm}\scriptsize $t_{j}^{(k,m)}$}
    \psfrag{t2j}{\hspace{-0.05cm}\scriptsize $t_{j}^{(k,m)}$}
    \psfrag{t2jp}{\hspace{-0.05cm}\scriptsize $t_{j^\prime}^{(k,m)}$}
    \psfrag{S2j}{\hspace{-0.15cm}\scriptsize $S_{j+1}^{(k,m)}$}
    \psfrag{S2jp}{\hspace{-0.1cm}\scriptsize $S_{j^\prime+1}^{(k,m)}$}
    \psfrag{S22}{\hspace{-0.05cm}\scriptsize $S_{3}^{(k,m)}$}
    \psfrag{W1j}{\hspace{-0.09cm}\scriptsize $W_{j}^{(k,m)}$}
    \psfrag{W2j}{\hspace{-0.25cm}\scriptsize $W_{j}^{(k,m)}$}
    \psfrag{Q1,j}{\hspace{-0.1cm}\scriptsize $Q_{j}^{(k,m)}$}
    \psfrag{Q2,j}{\hspace{-0.18cm}\scriptsize $Q_{j}^{(k,m)}$}
    \psfrag{r2,1}{\hspace{0.01cm}\scriptsize $\rho_1^{(k,m)}$}
    \psfrag{r2,2}{\hspace{-0.02cm}\scriptsize $\rho_2^{(k,m)}$}
    \psfrag{r2,j+1}{\hspace{0.02cm}\scriptsize $\rho_{j+1}^{(k,m)}$}
    \psfrag{r2,j}{\hspace{0.0cm}\scriptsize $\rho_j^{(k,m)}$}
    \psfrag{r2,j+1p}{\hspace{0.02cm}\scriptsize $\rho_{j^\prime+1}^{(k,m)}$}
    \psfrag{r2,jp}{\hspace{0.0cm}\scriptsize $\rho_{j^\prime}^{(k,m)}$}
    \psfrag{Beta12}{\hspace{0.05cm}\scriptsize $\beta_k^{(m)}$}
    \psfrag{Cross}{\hspace{0.0cm}\scriptsize Cross}
    \psfrag{point}{\hspace{0.0cm}\scriptsize point}
    }
    % \vspace{-0.15cm}
    \caption{Three conditions for packet drop via adaptive-length filtering.}
    \label{fig:fig3}
\end{figure}

The semantics-aware drop factor is denoted by $\psi_k^{(m)}$ for the link between the $k$-th SSM and $m$-th MM. Here, $\psi_k^{(m)}=1$ for the fixed-length filtering, while we can define
\begin{align}\label{sec4:eq1}
     \psi_k^{(m)} &=   \underset{T\rightarrow\infty}{\lim}\frac{1}{\mathcal{N}_k^{(m)}(T)} \sum_{j=1}^{\mathcal{N}_k^{(m)}(T)} \sum_{d=1}^{d_{\rm max}}
    \mathbbm{1}\!\left\{\dfrac{\rho_{j+d}^{(k,m)}}{\rho_{j}^{(k,m)}}
     \!>\! \tau^{(k,m)}_{d}\!\right\} \nonumber \\ 
     &\simeq  \underset{i \in \mathcal{I}_{n_k}}{\mathbb{E}} \!\Bigg[  \dfrac{1}{d_{\rm max}} \sum_{d=1}^{d_{\rm max}}
    {\rm Pr}\!\left\{i^\prime \in \mathcal{I}_{l_k^{(m)}} \bigg|\, \dfrac{\rho_{i^\prime}^{(k,m)}}{\rho_{i}^{(k,m)}}
     \!>\! \tau_d^{(k,m)}\!\right\}\!\Bigg]
\end{align}
for the adaptive-length filtering scheme.
In \eqref{sec4:eq1}, $\mathcal{N}_k^{(m)}(T)$ is the number of all admitted packets at the $k$-th SSM for the $m$-th MM by time $T$, and $d_{\rm max}\ll\mathcal{N}_k^{(m)}(T)$ for $T\!\rightarrow\!\infty$. Also, $\tau^{(k,m)}_{d}$ denotes a threshold for the drop of order $d$, which is computed in the following lemma.
\begin{lemma} \label{lemm1}
The threshold $\tau^{(k,m)}_{d}$ in \eqref{sec4:eq1} is derived as  
\begin{equation}\label{sec4:eq2}
    \tau^{(k,m)}_{d} = (d\!+\!1) + \dfrac{\widehat{W}_d^{(k,m)}}{\ell_{\rm max}}  
\end{equation}
where $\widehat{W}_d^{(k,m)}$ is an Erlang r.v. of order $d$ with rate $\lambda_kq_{l_k^{(m)}}$, and $\ell_{\rm max}$ indicates an upper bound to which the size of a codeword length converges. Notably, we reach $\tau^{(k,m)}_{d} = d+1$ if $\ell_{\rm max}\!\rightarrow\!\infty$.
\end{lemma}
\begin{proof} See Appendix~\ref{lemm:app1}.
\end{proof}

\subsection{Channel Error Control and Service Time Analysis}\label{Section4b}
We utilize two non-adaptive asynchronous error control schemes for handling transmission errors over block-fading channels, namely ARQ and truncated HARQ based on chase combining. In this regard, an SSM connected to one of the served MMs is equipped with an individual buffer to support ARQ and HARQ protocols. The successful delivery of each packet to an MM is declared by an instantaneous and error-free acknowledgment (ACK) feedback to the serving SSM (see Figure~\ref{fig:fig2}). However, in the event of failure at each monitor, the serving SSM retransmits the packet only to the MM from which a negative ACK (NACK) message is received.
Packets are retransmitted either until successful reception or up to the maximum allowable number of transmission rounds. After successful delivery to its destination or reaching the retransmission limit (in which case the packet is dropped), the transmitter waits for a new admitted arrival.\footnote{The computation and propagation delays are assumed negligible for both packet transmission and feedback processes.}

Consequently, the service time of packet $x_i^{(k)}$, which is tagged important for the $m$-th MM during the $j$-th arrival, is a function of its codeword length, the channel conditions, and the number of transmission rounds according to the error control protocol. Thus, we consider $\mathbb{E}[(S^{(k,m)})^c] = \mathbb{E}\!\left[\varphi^{(k,m)} (L^{(k,m)})^c\right]$ the $c$-th moment of the service time for packets being delivered to the $m$-th MM from the $k$-th SSM, where 
\begin{align}\label{sec2:eq2}
\mathbb{E}\!\left[\varphi^{(k,m)} (L^{(k,m)})^c\right] = \underset{T\rightarrow\infty}{\lim}\frac{1}{\mathcal{N}_k^{(m)}(T)} \sum_{j=1}^{\mathcal{N}_k^{(m)}(T)}  \sum_{r=1}^{r^{(k,m)}_{{\rm max}, j}} \varphi_{j,r}^{(k,m)} (\ell_j^{(k,m)})^c
    % &\simeq  \underset{i \in \mathcal{I}_{l_k^{(m)}}}{\mathbb{E}} \!\Bigg[ \dfrac{1}{r_{\rm max}} \sum_{r=1}^{r_{\rm max}}  \varphi_{i,r}^{(k,m)} (\ell_i^{(k,m)})^m \\
    % &~~\times {\rm Pr}\!\left\{i^\prime \in \mathcal{I}_{l_k^{(m)}} \bigg|\, \dfrac{\rho_i^{(k,m)}}{\rho_{i^\prime}^{(k,m)}} \!<\! \dfrac{1}{r + 1} \dfrac{\ell_{i^\prime}^{(k,m)}}{\ell_i^{(k,m)}}
    %  \!\right\}\!\Bigg]
\end{align}
since $f_k(\cdot)$ is non-increasing. Herein, $r$ is the order of transmission and truncated by $r^{(k,m)}_{{\rm max}, j} \ll \mathcal{N}_k^{(m)}(T)$, where $r^{(k,m)}_{{\rm max}, j}$ depends on the $j$-th arrival's meta-value, i.e., $v_j^{(k,m)}$. As a simple arbitrary form, we can write
\begin{align}\label{eq:rmaxeq}
    r_{{\rm max}, j}^{(k,m)} = \left(\!\dfrac{v_j^{(k,m)}} {\frac{1}{l_k^{(m)}} \sum_{i \in \mathcal{I}_{l_k^{(m)}}}\!v_i^{(k,m)}}\!\right)\! r_{\rm max}
\end{align}
for the $j$-th arrival. According to \eqref{eq:rmaxeq}, higher sample/packet importance results in larger $r^{(k,m)}_{{\rm max}, j}$. Besides, $\varphi_{j,r}^{(k,m)}$ is relevant to error control processes for the $j$-th arrival packet that can be transmitted up to $r$ times before reaching $r^{(k,m)}_{{\rm max}, j}$. In this regard, $\varphi_{j,r}^{(k,m)}$ is given by\cite{parag2017real}
\begin{align}\label{def1b}
    &\varphi_{j,r}^{(k,m)}=\dfrac{c_j^{(k,m)}}{1 - \theta_{r, j}^{(k,m)}} \!\left(r\theta_{r, j}^{(k,m)} \!+\! \sum_{r^\prime=1}^{r} r^\prime\! \left(\theta_{r^\prime-1, j}^{(k,m)} \!-\! \theta_{r^\prime, j}^{(k,m)}\right) \!\right)
\end{align}
where $c_j^{(k,m)}\geq 1$ indicates the reverse of the channel coding rate with $c_j^{(k,m)}=1$ for the ARQ protocol, and $\theta_{r, j}^{(k,m)}$ denotes the probability that the first $r$ transmissions of the $j$-th packet from the $k$-th SSM to the $m$-th MM are performed with error. 
For the chase combining HARQ, $\theta_{r,j}^{(k,m)}$, $\forall r,j,k,m\in\mathcal{M}(k)$, is given by \cite{lagrange2010throughput}
\begin{align}\label{theta1}
\theta_{r,j}^{(k,m)}=e^{-\dfrac{\gamma_{M,j}^{(k,m)}}{\bar{\gamma}_k^{(m)}}} \times\! \left(\sum_{i=r}^{\infty} \dfrac{\left(\!\dfrac{\gamma_{M,j}^{(k,m)}}{\bar{\gamma}_k^{(m)}}\!\right)^{\!\!i}}{i !} + \sum_{i=0}^{r-1} \dfrac{\left(\!\dfrac{\gamma_{M,j}^{(k,m)}}{\bar{\gamma}_k^{(m)}}\!\right)^{\!\!i}}{i !} \prod_{i^\prime=1}^{r-i} \dfrac{1}{1+i^\prime g_j^{(k,m)} \bar{\gamma}_k^{(m)}}\right),
\end{align}
while for ARQ we have
\begin{align}\label{theta2}
&\theta_{r,j}^{(k,m)}=\left(1 - \dfrac{g_j^{(k,m)}\bar{\gamma}_k^{(m)}}{1 + g_j^{(k,m)}\bar{\gamma}_k^{(m)}} e^{-\dfrac{\gamma_{M,j}^{(k,m)}}{\bar{\gamma}_k^{(m)}}}\right)^{\!\!r}.
\end{align}
Expressions \eqref{theta1} and \eqref{theta2} are obtained by applying an approximated expression for packet error rates as a function of the signal-to-noise ratio (SNR) for Rayleigh fading channels \cite{lagrange2010throughput, liu2004cross}, where $\bar{\gamma}_k^{(m)}$ is the average received SNR at the $m$-th MM from the $k$-th SSM. Fitting the approximated expression to the exact formula, we find $\gamma_{M,j}^{(k,m)}$ and $g_j^{(k,m)}$ under a considered modulation scheme.

\section{Semantic Source Coding}\label{Section5}
In this section, we formulate the problem of semantics-aware packet encoding, which aims to maximize the semantic value of packets delivered at the served monitors, and whose solution provides the optimal codeword lengths assigned to admitted packets at each sensor. 

\subsection{Problem Statement}
The objective of optimal codeword length assignment is to maximize the weighted sum of the average SoIs provided at all layers of MMs, subject to the following two constraints: (i) codeword lengths should be positive integers, i.e., $\ell_i^{(k,m)} \in \mathbb{Z}^+$ (feasibility); (ii) existence of a prefix-free (or uniquely decodable) code for a given set of codeword lengths at each SSM (Kraft-McMillan inequality \cite{cover1999elements}). Therefore, the optimization problem is formulated as
\begin{equation}\label{optim1}
\begin{aligned}
    &\mathcal{P}_1\!:\,
    \underset{{\{\ell_i^{(k,m)}\}}}{\text{maximize}}~~ \sum_{m=1}^{M} w_m \sum_{k \in \mathcal{K}(m)} p_k^{(m)}\bar{\mathcal{S}}_k^{(m)} \coloneqq \sum_{m=1}^{M} \sum_{k \in \mathcal{K}(m)} \mathcal{J}_{{\rm SoI}}^{(k,m)} \\
    &\text{subject to} \begin{cases} 
    \mathcal{C}_1: \sum_{i\in \mathcal{I}_{l_k^{(m)}}} 2^{-\ell_i^{(k,m)}} \leq 1, \forall k, m\\
    \mathcal{C}_2: \ell_i^{(k,m)} \in \mathbb{Z}^+, \forall i, k, m
    \end{cases}
\end{aligned}
\end{equation}
where $w_m$ denotes a weight parameter.

\begin{figure}[t!]
    \centering
    \subfloat[]{
    \pstool[scale=0.61]{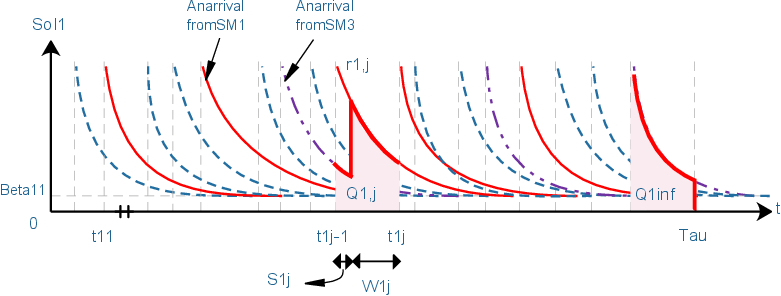}{
    \psfrag{SoI1}{\hspace{-0.2cm}\scriptsize $ \mathcal{S}_1^{(1)}(t)$}
    \psfrag{SoI2}{\hspace{-0.2cm}\tiny $ \mathcal{S}_2^{(1)}(t)$}
    \psfrag{t}{\hspace{0.0cm}\scriptsize $t$}
    \psfrag{0}{\hspace{0.0cm}\scriptsize $0$}
    \psfrag{t11}{\hspace{-0.1cm}\tiny $t_1^{(1,1)}$}
    \psfrag{t12}{\hspace{-0.1cm}\tiny $t_{1}^{(2,1)}$}
    \psfrag{t1j-1}{\hspace{-0.06cm}\tiny $t_{j-1}^{(1,1)}$}
    \psfrag{t2j-1}{\hspace{-0.06cm}\tiny $t_{j-1}^{(2,1)}$}
    \psfrag{t1j}{\hspace{-0.06cm}\tiny $t_{j}^{(1,1)}$}
    \psfrag{t2j}{\hspace{-0.06cm}\tiny $t_{j}^{(2,1)}$}
    \psfrag{S1j}{\hspace{-0.07cm}\tiny $S_{j}^{(1,1)}$}
    \psfrag{S2j}{\hspace{-0.07cm}\tiny $S_{j}^{(2,1)}$}
    \psfrag{W1j}{\hspace{-0.09cm}\tiny $W_{j}^{(1,1)}$}
    \psfrag{W2j}{\hspace{-0.25cm}\tiny $W_{j}^{(2,1)}$}
    \psfrag{Q1,j}{\hspace{-0.1cm}\tiny $Q_{j}^{(1,1)}$}
    \psfrag{Q2,j}{\hspace{-0.18cm}\tiny $Q_{j}^{(2,1)}$}
    \psfrag{r1,j}{\hspace{-0.05cm}\tiny $\rho_j^{(1,1)}$}
    \psfrag{r2,j}{\hspace{-0.00cm}\tiny $\rho_j^{(2,1)}$}
    \psfrag{Q1inf}{\hspace{-0.05cm}\tiny $Q_{\infty}^{(1,1)}$}
    \psfrag{Q2inf}{\hspace{-0.1cm}\tiny $Q_{\infty}^{(2,1)}$}
    \psfrag{Tau}{\hspace{0.06cm}\scriptsize $T$}
    \psfrag{Anarrival}{\hspace{-0.2cm}\scriptsize }
    \psfrag{fromSM1}{\hspace{-0.02cm}\scriptsize $\text{SSM}_1$}
    \psfrag{fromSM3}{\hspace{-0.15cm}\scriptsize $\text{SSM}_3$}
    \psfrag{Beta11}{\hspace{-0.03cm}\scriptsize $\beta_1^{(1)}$}
    \psfrag{Beta12}{\hspace{-0.02cm}\scriptsize $\beta_2^{(1)}$}
    }}
    % \hfil
    \subfloat[]{    
    \pstool[scale=0.61]{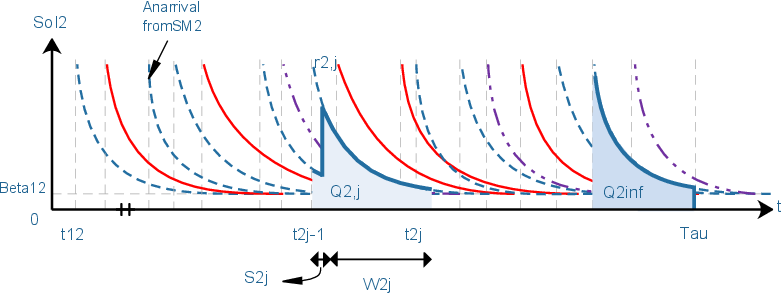}{
    \psfrag{SoI1}{\hspace{-0.2cm}\scriptsize $ \mathcal{S}_1^{(1)}(t)$}
    \psfrag{SoI2}{\hspace{-0.2cm}\scriptsize $ \mathcal{S}_2^{(1)}(t)$}
    \psfrag{t}{\hspace{0.0cm}\scriptsize $t$}
    \psfrag{0}{\hspace{0.0cm}\scriptsize $0$}
    \psfrag{t11}{\hspace{-0.1cm}\tiny $t_1^{(1,1)}$}
    \psfrag{t12}{\hspace{-0.1cm}\tiny $t_{1}^{(2,1)}$}
    \psfrag{t1j-1}{\hspace{-0.06cm}\tiny $t_{j-1}^{(1,1)}$}
    \psfrag{t2j-1}{\hspace{-0.06cm}\tiny $t_{j-1}^{(2,1)}$}
    \psfrag{t1j}{\hspace{-0.06cm}\tiny $t_{j}^{(1,1)}$}
    \psfrag{t2j}{\hspace{-0.06cm}\tiny $t_{j}^{(2,1)}$}
    \psfrag{S1j}{\hspace{-0.07cm}\tiny $S_{j}^{(1,1)}$}
    \psfrag{S2j}{\hspace{-0.07cm}\tiny $S_{j}^{(2,1)}$}
    \psfrag{W1j}{\hspace{-0.09cm}\tiny $W_{j}^{(1,1)}$}
    \psfrag{W2j}{\hspace{-0.25cm}\tiny $W_{j}^{(2,1)}$}
    \psfrag{Q1,j}{\hspace{-0.1cm}\tiny $Q_{j}^{(1,1)}$}
    \psfrag{Q2,j}{\hspace{-0.18cm}\tiny $Q_{j}^{(2,1)}$}
    \psfrag{r1,j}{\hspace{-0.05cm}\tiny $\rho_j^{(1,1)}$}
    \psfrag{r2,j}{\hspace{-0.00cm}\tiny $\rho_j^{(2,1)}$}
    \psfrag{Q1inf}{\hspace{-0.05cm}\tiny $Q_{\infty}^{(1,1)}$}
    \psfrag{Q2inf}{\hspace{-0.1cm}\tiny $Q_{\infty}^{(2,1)}$}
    \psfrag{Tau}{\hspace{0.06cm}\scriptsize $T$}
    \psfrag{Anarrival}{\hspace{-0.2cm}\scriptsize }
    \psfrag{fromSM2}{\hspace{-0.1cm}\scriptsize $\text{SSM}_2$}
    \psfrag{Beta11}{\hspace{-0.03cm}\scriptsize $\beta_1^{(1)}$}
    \psfrag{Beta12}{\hspace{-0.02cm}\scriptsize $\beta_2^{(1)}$}
    }}
    \vspace{-0.2cm}
    \caption{Sample evolution for the EUT case, where (a) $\text{MM}_1$ receives updates from identical $\text{SSM}_1$ and $\text{SSM}_3$, and (b) $\text{MM}_2$ receives packets from $\text{SSM}_2$.}
    \label{fig:fig4}
\end{figure}
To derive $\bar{\mathcal{S}}_k^{(m)}$, we use different forms of $f_k(\cdot)$ defined in \eqref{sec5:eq3}, considering the type and functionality of the $k$-th sensor. Thus, from \eqref{sec3b:eq2} and \eqref{sec5:eq3}, the average SoI can be computed for the proposed cases. To do so, we divide the non-negative area below the curve of $f_k(\Delta_k^{(m)}(t))$ over interval $(0,T)$ into polygons of $Q_{j}^{(k,m)}$, $j=1, 2, ..., \mathcal{N}_k^{(m)}(T)$, and $Q_{\infty}^{(k,m)}$, as depicted in Figure~\ref{fig:fig4} for the EUT case. Thus, we can rewrite \eqref{sec3b:eq2} as
\begin{eqnarray}\label{sec5:eq5}
    \bar{\mathcal{S}}_k^{(m)}=\underset{T\rightarrow\infty}{\lim} \dfrac{1}{T}\! \left(\sum_{j=1}^{\mathcal{N}_k^{(m)}(T)} Q_{j}^{(k,m)} \!+ Q_{\infty}^{(k,m)} \!\right)\!
    = \eta_k^{(m)} \mathbb{E}[Q^{(k,m)}]
\end{eqnarray}
where $\eta_k^{(m)} = \underset{T\rightarrow\infty}{\lim}\dfrac{1}{T}\!\left(\mathcal{N}_k^{(m)}(T)\!-\!1\right)$ denotes the steady-state time average arrival rate. Importing \eqref{sec5:eq5} into \eqref{optim1}, we reach the following problem.
\begin{equation}\label{optim2}
\begin{aligned}
    &\mathcal{P}_2\!:\,
    \underset{{\{\ell_i^{(k,m)}\}}}{\text{maximize}}~~ \sum_{m=1}^{M} w_m \sum_{k \in \mathcal{K}(m)} p_k^{(m)}\eta_k^{(m)} \mathbb{E}[Q^{(k,m)}]\\
    &\text{subject to} \begin{cases} 
    \mathcal{C}_1: \sum_{i\in \mathcal{I}_{l_k^{(m)}}} 2^{-\ell_i^{(k,m)}} \leq 1, \forall k, m\\
    \hat{\mathcal{C}}_2: \ell_i^{(k,m)} \in \mathbb{R}^+, \forall i, k, m.
    \end{cases}
\end{aligned}
\end{equation}
where constraint $\hat{\mathcal{C}}_2$ is a relaxation of $\mathcal{C}_2$ in \eqref{optim1} to allow for non-negative real-valued codeword lengths. The solution of $\mathcal{P}_2$ gives real-valued codeword lengths, whose corresponding integer values can be found using a rounded-off operation.

\subsection{Semantics-Aware Encoding Design}\label{Section5b}
To solve the problem $\mathcal{P}_2$, in what follows, we find $\mathbb{E}[Q^{(k,m)}]$ based on different forms of $f_k(\Delta_k^{(m)}(t))$ defined in \eqref{sec5:eq3}.

\subsubsection{EUT case}
Considering the exponential utility function of timeliness, we propose the following lemma.
\begin{lemma}\label{lemm2}
The expected form of polygons $Q_{j}^{(k,m)}$ for the EUT case is approximately derived as
\begin{align}\label{sec5:eq8}
    \mathbb{E}[Q^{(k,m)}] &\approx 
    \Big( 1\!+\!\beta_k^{(m)} \!-\! \bar{\rho}_k^{(m)}\gamma_k^{(m)} \!-\! \mathbb{E}\!\left[\rho^{(k,m)}\varphi^{(k,m)}L^{(k,m)}\right]\! \Big)\mathbb{E}\!\left[\varphi^{(k,m)}L^{(k,m)}\right]
    - \dfrac{1}{2} \bar{\rho}_k^{(m)}\mathbb{E}\!\left[\varphi^{(k,m)}(L^{(k,m)})^2\right] \nonumber \\
    &~~~
   - \gamma_k^{(m)} \mathbb{E}\!\left[\rho^{(k,m)}\varphi^{(k,m)}L^{(k,m)}\right] - \bar{\rho}_k^{(m)}(\gamma_k^{(m)})^2 + (1\!+\!\beta_k^{(m)})\gamma_k^{(m)}
\end{align}
where $\gamma_k^{(m)} \coloneqq 1/(\lambda_k \psi_k^{(m)} q_{l_k^{(m)}})$ and $\bar{\rho}_k^{(m)} \coloneqq \mathbb{E}[\rho^{(k,m)}] = \displaystyle  \underset{T\rightarrow\infty}{\lim} \dfrac{1}{T} \sum_{j=1}^{\mathcal{N}_k^{(m)}(T)} \rho_j^{(k,m)} = \sum_{i\in \mathcal{I}_{l_k^{(m)}}} p_i^{(k,m)} \rho_i^{(k,m)}$.
\end{lemma}
\begin{proof} See Appendix~\ref{lemm:app2}.
\end{proof}

Importing \eqref{sec5:eq8} into \eqref{optim2}, $\mathcal{P}_2$ becomes a convex problem and can be solved using standard solvers, e.g., MOSEK in CVX. Alternatively, we propose a heuristic solution to compute closed-form expressions for the codeword length, as follows.
\begin{proposition}\label{prop1}
The codeword lengths $\ell_i^{(k,m)}$, $\forall i\in \mathcal{I}_{l_k^{(m)}}$, which maximize \eqref{optim2} in $\mathcal{P}_2$ for the EUT case are computed as 
\begin{align}\label{sec5:eq9}
    \ell_i^{(k,m)} =-\log_2\!\left(\! \dfrac{\bar{\rho}_k^{(m)} p_i^{(k,m)} \varphi_i^{(k,m)}}{\mu_k^{(m)} (\ln(2))^2} \mathcal{W}_0\!\left(\! \dfrac{\mu_k^{(m)} (\ln(2))^2}{\bar{\rho}_k^{(m)} p_i^{(k,m)} \varphi_i^{(k,m)}} 2^{{\xi_k^{(m)}}}\!\right) \!\right)\!
\end{align}
where $\mu_k^{(m)}\geq0$ is a constant multiplier,  
\begin{align}\label{sec5:eq10}
    \xi_k^{(m)} \coloneqq \dfrac{2\chi_k^{(m)}\mu_k^{(m)}\ln(2) + \bar{\rho}_k^{(m)}\gamma_k^{(m)}\big(1\!+\!\chi_k^{(m)}\big) - \big(1\!+\!\beta_k^{(m)}\big)}{\bar{\rho}_k^{(m)}\big(1 \!+\! 2 \chi_k^{(m)}\bar{\varphi}_k^{(m)}\big)},
\end{align}
$\bar{\varphi}_k^{(m)} \coloneqq \mathbb{E}[\varphi^{(k,m)}] = \underset{T\rightarrow\infty}{\lim} \dfrac{1}{T} \sum_{j=1}^{\mathcal{N}_k^{(m)}(T)} \varphi_j^{(k,m)}$, and
\begin{equation}\label{sec5:eq11}
   \chi_k^{(m)} =  \dfrac{\mathbb{E}\!\left[\rho^{(k,m)}\varphi^{(k,m)}L^{(k,m)}\right]}{\bar{\rho}_k^{(m)} \mathbb{E}\!\left[\varphi^{(k,m)}L^{(k,m)}\right]}.
\end{equation}
\end{proposition}
\begin{proof} See Appendix~\ref{app1}.
\end{proof}

The optimal codeword lengths\footnote{The optimal codeword lengths refer to \emph{real-valued} optimal codeword lengths throughout the paper.} of delivered packets from the $k$-th SSM to the $m$-th MM can be found via \eqref{sec5:eq9} for $(\mu_k^{(m)}, \chi_k^{(m)})$ pair, subject to $\mu_k^{(m)}>0$ and $\sum_{i\in \mathcal{I}_{l_k^{(m)}}} 2^{-\ell_i^{(k,m)}}=1$. The values of $\chi_k^{(m)}$ and $\mu_k^{(m)}$ can be obtained using Algorithm~\ref{Alg:Alg.1}, through its \emph{inner} and \emph{outer} loops, respectively. We first assume similar importance, hence attenuation factors, for all packets and initialize $\chi_k^{(m)} = 1$. Then, given a small value of $\mu_k^{(m)}$, we compute $\ell_i^{(k,m)}$ and new $\chi_k^{(m)}$ through the inner loop. Thereafter, based on the found $\chi_k^{(m)}$, we derive new values for $\ell_i^{(k,m)}$ in the next iteration. This process continues until we reach the stopping accuracy $\varepsilon$. Once reached, the outer loop checks the Kraft-McMillan condition and resets $\mu_k^{(m)}$ if the condition is not satisfied. Subsequently, the inner loop starts again, considering the new value of the constant multiplier. Finally, the algorithm converges to the final amounts of $\chi_k^{(m)}$ and $\mu_k^{(m)}$ with the rate of $\mathcal{O}((N_aN_b)^{-1})$ in which $N_a$ and $N_b$ denote the maximum numbers of outer and inner iterations, respectively.
\begin{algorithm}[t!]
\DontPrintSemicolon
    \caption{Solution for deriving $\mu_k^{(m)}$ and $\chi_k^{(m)}$} \label{Alg:Alg.1}
    \KwInput{Fixed parameters $\mathcal{I}_{l_k^{(m)}}$, $p_i^{(k,m)}$, $\rho_i^{(k,m)}$, and $\varphi_i^{(k,m)}$, $\forall i \in \mathcal{I}_{l_k^{(m)}}$, and $\beta_k^{(m)}$.
    Stopping accuracy $\varepsilon$. Initial parameters $(\mu_k^{(m)})^{(0)}$, $(\chi_k^{(m)})^{(0)}$, $(\xi_k^{(m)})^{(0)}$, and $(\ell_i^{(k,m)})^{(0)}$, $\forall i \in \mathcal{I}_{l_k^{(m)}}$.
    }
    \KwOutput{Computed parameters  $\chi_k^{(m)}=(\chi_k^{(m)})^{(b)}$, $\xi_k^{(m)} = (\xi_k^{(m)})^{(b)}$, $\ell_i^{(k,m)}=(\ell_i^{(k,m)})^{(b)}$, $\forall i \in \mathcal{I}_{l_k^{(m)}}$, and $\mu_k^{(m)}=(\mu_k^{(m)})^{(a)}$.}
    \textit{Iteration} $a$: \Comment{Outer loop}\\
    \textit{Iteration} $b$: \Comment{Inner loop}\\
     Compute $(\xi_k^{(m)})^{(b)}$ and $(\ell_i^{(k,m)})^{(b)}$,$\forall i \in \mathcal{I}_{l_k^{(m)}}$, using \eqref{sec5:eq10} and \eqref{sec5:eq9}, respectively.\\
     Calculate $\mathbb{E}\!\left[\rho^{(k,m)}\varphi^{(k,m)}L^{(k,m)}\right]$ and $\bar{\rho}_k^{(m)} \mathbb{E}\!\left[\varphi^{(k,m)}L^{(k,m)}\right]$.\\
     Update $(\chi_k^{(m)})^{(b)}$ from \eqref{sec5:eq11} based on {\scriptsize{\textbf{4}}}.\\
     \lIf{Criterion $\big|(\chi_k^{(m)})^{(b)} \!-\! (\chi_k^{(m)})^{(b-1)}\big|\!>\! \varepsilon$}{set $b=b+1$, and \textbf{goto} {\scriptsize{\textbf{2}}}.}
     Update $(\xi_k^{(m)})^{(b)}$ and $(\ell_i^{(k,m)})^{(b)}$ according to {\scriptsize{\textbf{5}}}.\\
     \lIf{$\sum_{i\in \mathcal{I}_{l_k^{(m)}}}2^{-(\ell_i^{(k,m)})^{(b)}}=1$}{stop the process, and \textbf{goto} {\scriptsize{\textbf{11}}}.}
     \lElseIf{$\sum_{i\in \mathcal{I}_{l_k^{(m)}}}2^{-(\ell_i^{(k,m)})^{(b)}}<1$}{decrease $(\mu_k^{(m)})^{(a)}$, set $a=a+1$, and \textbf{goto} {\scriptsize{\textbf{1}}}.}
     \lElse{increase $(\mu_k^{(m)})^{(a)}$, set $a=a+1$, and \textbf{goto} {\scriptsize{\textbf{1}}}.}
     Save $(\chi_k^{(m)})^{(b)}$, $(\xi_k^{(m)})^{(b)}$, $(\ell_i^{(k,m)})^{(b)}$, $\forall i \in \mathcal{I}_{l_k^{(m)}}$, and $(\mu_k^{(m)})^{(a)}$.
\end{algorithm}

\subsubsection{LUT case} Under a logarithmic utility function of timeliness, $\mathbb{E}[Q^{(k,m)}]$ is found as follows.
\begin{lemma}\label{lemm3}
The expected value of $Q_{j}^{(k,m)}$ for the LUT case is approximately derived as
\begin{align}\label{sec5:eq13} 
    \mathbb{E}[Q^{(k,m)}] &\approx 
    - 2\mathbb{E}\!\left[\rho^{(k,m)}\varphi^{(k,m)}L^{(k,m)}\right]\mathbb{E}\!\left[\varphi^{(k,m)}L^{(k,m)}\right] -  \bar{\rho}_k^{(m)}\mathbb{E}\!\left[\varphi^{(k,m)}(L^{(k,m)})^2\right] \nonumber \\ &
   ~~~ - 2\gamma_k^{(m)} \mathbb{E}\!\left[\rho^{(k,m)}\varphi^{(k,m)}L^{(k,m)}\right] + \Big( \beta_k^{(m)} \!-\! 1 \!-\! 2\bar{\rho}_k^{(m)}\gamma_k^{(m)}\Big)\mathbb{E}\!\left[\varphi^{(k,m)}L^{(k,m)}\right] \nonumber \\ 
   &~~~ - 2\bar{\rho}_k^{(m)}(\gamma_k^{(m)})^2 + (\beta_k^{(m)} \!-\! 1)\gamma_k^{(m)}.
\end{align}
\end{lemma}
\begin{proof} See Appendix~\ref{lemm:app3}.
\end{proof}

Inserting \eqref{sec5:eq13} into \eqref{optim2}, $\mathcal{P}_2$ is a convex problem, and the optimal codeword lengths $\ell_i^{(k,m)}$ can be obtained using either standard solvers or the following expression
\begin{align}\label{sec5:eq14}
    \ell_i^{(k,m)} =-\log_2\!\left(\! \dfrac{\bar{\rho}_k^{(m)} p_i^{(k,m)} \varphi_i^{(k,m)}}{\hat{\mu}_k^{(m)} (\ln(2))^2} \mathcal{W}_0\!\left(\! \dfrac{\hat{\mu}_k^{(m)} (\ln(2))^2}{\bar{\rho}_k^{(m)} p_i^{(k,m)} \varphi_i^{(k,m)}} 2^{{\hat{\xi}_k^{(m)}}}\!\right) \!\right)\!,
\end{align}
which is derived with the same method as in \eqref{sec5:eq9}. Herein, $\hat{\mu}_k^{(m)}\geq0$ indicates a constant. We also have $\hat{\xi}_k^{(m)} \coloneqq \dfrac{4\hat{\chi}_k^{(m)}\hat{\mu}_k^{(m)}\ln(2) + 2\bar{\rho}_k^{(m)}\gamma_k^{(m)}\big(1\!+\!\hat{\chi}_k^{(m)}\big) - \big(\beta_k^{(m)} \!-\!1\big)}{2\bar{\rho}_k^{(m)}\big(1 \!+\! 2 \hat{\chi}_k^{(m)}\bar{\varphi}_k^{(m)}\big)}$.
The optimal codeword lengths are computed in \eqref{sec5:eq14} based on the known pair of $(\hat{\mu}_k^{(m)}, \hat{\chi}_k^{(m)})$, subject to $\hat{\mu}_k^{(m)}>0$ and $\sum_{i\in \mathcal{I}_{l_k^{(m)}}} 2^{-\ell_i^{(k,m)}}=1$. Similar to the EUT case, we can find the values of $\hat{\mu}_k^{(m)}$ and $\hat{\chi}_k^{(m)}$ using Algorithm~\ref{Alg:Alg.1} replacing its parameters with those for the LUT case.

\subsubsection{RUT case} Following the same steps as for deriving \eqref{sec5:eq8} and \eqref{sec5:eq13}, the expected form of polygons $Q_{j}^{(k,m)}$ for the RUT case is obtained as
\begin{align}\label{sec5:eq16}
    \mathbb{E}[Q^{(k,m)}] &\approx 
    - \kappa \mathbb{E}\!\left[\rho^{(k,m)}\varphi^{(k,m)}L^{(k,m)}\right]\mathbb{E}\!\left[\varphi^{(k,m)}L^{(k,m)}\right] - \frac{\kappa}{2} \mathbb{E}\!\left[\rho^{(k,m)}\varphi^{(k,m)}(L^{(k,m)})^2\right] \nonumber \\
    &~~~ - 2 \kappa \gamma_k^{(m)} \mathbb{E}\!\left[\rho^{(k,m)}\varphi^{(k,m)}L^{(k,m)}\right] + \Big( \frac{\kappa}{\kappa \!+\!1}\!+\!\beta_k^{(m)} \!-\! \kappa \bar{\rho}_k^{(m)}\gamma_k^{(m)}\Big)\mathbb{E}\!\left[\varphi^{(k,m)}L^{(k,m)}\right] \nonumber \\
    &~~~ -2 \kappa \bar{\rho}_k^{(m)}(\gamma_k^{(m)})^2 + (\frac{\kappa}{\kappa \!+\!1}\!+\!\beta_k^{(m)})\gamma_k^{(m)}.
\end{align}
Consequently, the optimal codeword lengths are given as
\begin{align}\label{sec5:eq17}
    \ell_i^{(k,m)} =-\log_2\!\left(\! \dfrac{\bar{\rho}_k^{(m)} p_i^{(k,m)} \varphi_i^{(k,m)}}{\doublehat{\mu}_k^{(m)} (\ln(2))^2} \mathcal{W}_0\!\left(\! \dfrac{\doublehat{\mu}_k^{(m)} (\ln(2))^2}{\bar{\rho}_k^{(m)} p_i^{(k,m)} \varphi_i^{(k,m)}} 2^{\hat{\hat{\xi}}_k^{(m)}}\!\right) \!\right)\!
\end{align}
where $\doublehat{\mu}_k^{(m)}\geq0$, and $  \doublehat{\xi}_k^{(m)} \coloneqq \dfrac{2\kappa\doublehat{\chi}_k^{(m)}\doublehat{\mu}_k^{(m)}\ln(2) + \kappa \bar{\rho}_k^{(m)}\gamma_k^{(m)}\big(1\!+\!2\doublehat{\chi}_k^{(m)}\big) - \big(\frac{\kappa}{\kappa\!+\!1}\!+\!\beta_k^{(m)}\big)}{\kappa \bar{\rho}_k^{(m)}\big(1 \!+\! 2 \doublehat{\chi}_k^{(m)}\bar{\varphi}_k^{(m)}\big)}$.
The values of $\doublehat{\mu}_k^{(m)}$ and $\doublehat{\chi}_k^{(m)}$ are calculated using Algorithm~\ref{Alg:Alg.1} with parameters for the RUT case.

\subsection{Asymptotic Expansions for Codeword Lengths}
We provide here two asymptotic expansions for the derived closed-form codeword lengths in \eqref{sec5:eq9}, \eqref{sec5:eq14}, and \eqref{sec5:eq17}, for the EUT, LUT, and RUT cases, respectively.

\subsubsection{Dependency on occurrence probability and transmission rounds}
As shown in Section~\ref{Section5b}, the codeword length assigned to a realization (update packet) depends - among others - on its probability of occurrence and the number of transmission rounds. Using
\begin{equation}
    \left[\frac{\mathcal{W}_0(y)}{y}\right]^z = \operatorname{exp}\!\big(\!-z \mathcal{W}_0(y)\big), \forall y\geq e,
\end{equation}
from the Laurent series of order $z$, and $\mathcal{W}_0(y) = \ln(y) - \ln(\ln(y)) + \mathcal{O}(1)$ from the second-order Taylor expansion for large enough $y$, we can write
\begin{equation}\label{ass:eq1}
    \ell_i^{(k,m)} \propto \dfrac{\underset{i }{\operatorname{max}}\Big\{p_i^{(k,m)}\varphi_i^{(k,m)}\Big\}}{p_i^{(k,m)}\varphi_i^{(k,m)}} - \ln\!\left(\dfrac{\underset{i }{\operatorname{max}}\Big\{p_i^{(k,m)}\varphi_i^{(k,m)}\Big\}}{p_i^{(k,m)}\varphi_i^{(k,m)}}\!\right)\!.
\end{equation}
From the asymptotic expression in \eqref{ass:eq1}, we deduce that the codeword length monotonically decreases (increases) by an increase in the occurrence probability, the number of transmissions, or both when $p_i^{(k,m)} \varphi_i^{(k,m)} \leq  1$ ($p_i^{(k,m)} \varphi_i^{(k,m)} >  1$). Besides, $p_i^{(k,m)} \varphi_i^{(k,m)} \!\rightarrow 0$ yields $\ell_i^{(k,m)}\!\rightarrow\ell_{\rm max}$, $\forall i \in \mathcal{I}_{l_k^{(m)}}$, where $\ell_{\rm max}$ denotes the upper bound for the size of a codeword length. 

\subsubsection{Uniform sources and equal semantic value}
For a uniform source, i.e., $p_i^{(k,m)}=p_k^{(m)}/n_k$, $\forall i$, with equal goal-oriented importance, e.g., $v_i^{(k,m)} = 1$, and the same number of transmissions, i.e., $\varphi_i^{(k,m)} = \varphi_{i^\prime}^{(k,m)}$, $\forall i \neq i^\prime$, we end up with codewords of equal size, i.e., $\ell_i^{(k,m)}=\ell_{i^\prime}^{(k,m)}$, where
\begin{equation}\label{ass:eq2}
    \ell_i^{(k,m)} \propto \dfrac{n_k}{p_k^{(m)}} - \ln\!\left(\dfrac{n_k}{p_k^{(m)}}\!\right)\!, \forall i \in \mathcal{I}_{l_k^{(m)}}.
\end{equation}
Thus, $n_k\gg1$ results in $\ell_i^{(k,m)}\!\rightarrow\ell_{\rm max}$, which remains almost fixed for very large $n_k$. Therefore, as the number of realizations \emph{increases}, the assigned codeword lengths become \emph{longer} to satisfy the Kraft-McMillan condition at the expense of \emph{higher} service time. On the other hand, \eqref{ass:eq2} indicates that the codeword lengths of packets transferred to a monitor with a \emph{higher} probability of observation from a serving sensor are \emph{shorter} than those with a lower probability of observation from the same sensor. That is, $\ell_i^{(k,m)}$ of packets sent to the $m$-th SSM from the $k$-th SSM become shorter with the increase of $p_k^{(m)}$.

\section{Simulation Results}\label{Section6}
In this section, we present simulation results that corroborate our analysis and show the performance gains by properly designing semantic filtering and source coding for timely status update delivery in DMSs.

\subsection{Setup and Assumptions}
We consider a DMS with $K=100$ randomly distributed SSMs and $M=16$ fixedly positioned MMs in an area of $10\times10\, [\text{km}^2]$ divided into sixteen identical subareas. %(see Figure~\ref{fig:map}). 
The $k$-th SSM transmits update packets to its four closest MMs, i.e., $|\mathcal{M}(k)|=4$, $\forall k$, over Rayleigh block-fading channels. The $m$-th MM observes packets from its serving SSMs with the same probabilities, i.e., $p_k^{(m)} = \frac{1}{|\mathcal{K}(m)|}$, $\forall m, k\in \mathcal{I}_{l_k^{(m)}}$.\footnote{The rationale behind the uniform assumption is twofold; first for its simplicity, and second for its relevance in practical use cases in which sensor and monitor modules do not have accurate knowledge of other agents’ locations and activation patterns.} The information source follows a ${\rm Zipf}(n^{(k)},s)$ distribution with pmf $P_X(x_i^{(k)})= \frac{1/i^s}{\sum_{l=1}^{n^{(k)}} {1}/{{l}^s}}$ for $n_k = | \mathcal{X}^{(k)} |$ realizations at the $k$-th sensor. The parameter $s$ is an exponent characterizing $P_X(\cdot)$,varying from uniform distribution ($s=0$) to \emph{peaky} ones. 
\noindent
\renewcommand{\arraystretch}{1}
\begin{table}[!ht]
\begin{center}
\caption{Intrinsic and Extrinsic features with their proposed properties.}\label{tab:tab2}
\begin{tabular}{| l | l | l | l | l | }
\hline
\multicolumn{1}{|c|}{\rule{0pt}{9pt} \small \centering \!\!\text{Feature}}& \multicolumn{1}{c|}{\rule{0pt}{9pt} \small \centering \!\!\!\!\!$w_b^{(k,m)}$\!\!\!} & \multicolumn{1}{c|}{\rule{0pt}{9pt} \small \centering \!\!\!\!\!$U_n^{(k,m)}$\!\!\!} & \multicolumn{1}{c|}{\rule{0pt}{9pt} \centering \small \!\!\text{Limits}} & \multicolumn{1}{c|}{\rule{0pt}{9pt} \small \centering\!\!Critical points}\\
\hline
\hline
\footnotesize \!Probability & \footnotesize -- & \footnotesize -- & \footnotesize \!\!$(0, 1]$\!\! & \footnotesize \!\!$\{0\}$ \\
\hline
\footnotesize \!Usefulness & \footnotesize -- & \footnotesize -- & \footnotesize \!\!$[1, n_k]$\!\! & \footnotesize \!\!Ten random $x_i^{(k)}$ \\
\hline
\footnotesize \!Loss risk* & \footnotesize -- & \footnotesize -- & \footnotesize $\!\![0, 100]\%\!\!$ & \footnotesize \!\!$\{100\}\%$ \\
\hline
\footnotesize \!Average distance & \footnotesize\!\!$0.4$\!\! & \footnotesize\!\!$0.4$\!\! & \footnotesize \!\!$[3.66, 8.68)\,[\text{km}]$\!\! & \footnotesize \!\!$\{3.66\}\,[\text{km}]$ \\
\hline
\footnotesize \!Resolution & \footnotesize \!\!$0.2$\!\! & \footnotesize\!\!$0.6$\!\! &  \footnotesize \!\!$(0, R_{\rm max}]$\!\! & \footnotesize\!\!$\{R_{\rm max}\}$\!\!\\
\hline
\footnotesize \!Tolerance & \footnotesize \!\!$0.2$\!\! & \footnotesize\!\!$1$\!\! & \footnotesize$\!\![0, 100]\%\!\!$ & \footnotesize \!\!$\{100\}\%$\!\!\\
\hline
\footnotesize \!Circuit power & \footnotesize \!\!$0.2$\!\! & \footnotesize\!\!$0.63$\!\! & \footnotesize\!\!$[P_{\rm min}, \infty)$\!\! & \footnotesize \!\!$\{P_{\rm min}\}$\!\!\\
\hline
\footnotesize \!Battery state  & \footnotesize \!\!$0.3$\!\! & \footnotesize\!\!$1$\!\! &\footnotesize $\!\![0, 100]\%\!\!$ & \footnotesize  \!\!$\{0\}\%$ \\
\hline
\multicolumn{5}{|l|}{\footnotesize \!*The criticality for all critical points is equal to $1$.} \\
\hline
\end{tabular}
\medskip
\end{center}
\end{table}
\noindent
\renewcommand{\arraystretch}{1}
\begin{table*}[!t]
\begin{center}
\caption{Parameters for simulation results.}\label{tab:tab3}
\begin{tabular}{ | l | c | l || l | c | l | }
\hline
% \rowcolor{yellow}
\multicolumn{1}{|c|}{\rule{0pt}{9pt} \small \centering \text{Parameter}} & \small \centering \text{Symbol} & \multicolumn{1}{c||}{\centering \small \!\text{Value}\!} & \multicolumn{1}{c|}{\rule{0pt}{9pt} \small \centering \text{Parameter}} & \small \centering \text{Symbol} & \multicolumn{1}{c|}{\centering \small \text{Value}}\\
\hline
\hline
\footnotesize \!Number of the SSMs & \footnotesize $K$ & \footnotesize \!\!$100$ & \footnotesize \!Minimum attenuation factors & \footnotesize $\rho_{\rm min}$ & \footnotesize \!\!$0.1$\\
\hline
\footnotesize \!Number of the MMs & \footnotesize $M$ & \footnotesize \!\!$16$ & \footnotesize \!Maximum attenuation factors & \footnotesize $\rho_{\rm max}$ & \footnotesize \!\!$5$\\
\hline
\footnotesize \!Size of realizations set & \footnotesize $n_k$, $\forall k$ & \footnotesize \!\!$100$ & \footnotesize \!Maximum packet drop times & \footnotesize $d_{\rm max}$ & \footnotesize \!\!$10$\\
\hline
\footnotesize \!Constant exponent for ${\rm Zipf}(.,.)$\!& \footnotesize $s$ & \footnotesize \!\!$0.4$ & \footnotesize \!Maximum transmission rounds\! & \footnotesize $r_{\rm max}$ & \footnotesize \!\!$3$ \cite{liu2004cross}\\
\hline
\footnotesize \!Update packet input rate & \footnotesize $\lambda_k$, $\forall k$ & \footnotesize \!\!$0.5$ & \footnotesize \!Reverse of channel coding rate\!& \footnotesize \!\!$c_j^{(k,m)}$, $\forall k,b,j$\!\! & \footnotesize \!\!$2$ \cite{liu2004cross}\\
\hline
\footnotesize \!Number of intrinsic features & \footnotesize $A$ & \footnotesize \!\!$3$ & \footnotesize \!Constant weight parameter & \footnotesize $w_m$, $\forall m$ & \footnotesize \!\!$1$\\
\hline
\footnotesize \!Number of extrinsic features & \footnotesize $B$ & \footnotesize \!\!$3$ & \footnotesize \!Upper bound of codeword lengths & \footnotesize $\ell_{\rm max}$ & \footnotesize \!\!$100$\\
\hline
\footnotesize \!Exponent for intrinsic features & \footnotesize \!\!$\alpha_a^{(k,m)}$, $\forall a,k,m$\!\! & \footnotesize \!\!$0.5$ & \footnotesize \!Monitoring time length of arrivals\! & \footnotesize $T$ & \footnotesize \!\!$1000\,[\text{sec}]$\\
\hline
\footnotesize \!Exponent for extrinsic features & \footnotesize \!\!$\bar{\alpha}_b^{(k,m)}$, $\forall b,k,m$\!\! & \footnotesize \!\!$1$ & \footnotesize \!Minimum circuit power & \footnotesize $P_{\rm min}$ & \footnotesize \!\!$0.1\, [\text{W}]$\\
\hline
\footnotesize \!Minimum importance & \footnotesize $\operatorname{VF}_{\rm min}$ & \footnotesize \!\!$0.1$ & \footnotesize \!Maximum sensing resolution & \footnotesize $R_{\rm max}$ & \footnotesize \!\!$5\, [\text{m}]$\\
\hline
% \footnotesize \!Threshold for the hybrid extension & \footnotesize $v_{\tau}^{(k,m)}$, $\forall k,b$ & \footnotesize \!\!$0.1$\\
% \hline
\footnotesize \!Bias value for the utility forms & \footnotesize $\beta_k^{(m)}$, $\forall k,m$ & \footnotesize \!\!$5$ & \footnotesize \!Average received SNR & \footnotesize $\bar{\gamma}_k^{(m)}$, $\forall k,b$ & \footnotesize \!\!$12\, [\text{dB}]$\\
\hline
\multirow{2}{*}{\footnotesize \!Exponent for the RUT case} & \multirow{2}{*}{\footnotesize $\kappa$} & \multirow{2}{*}{\footnotesize \!\!$2$} & \multirow{2}{*}{\footnotesize \begin{tabular}{@{}l@{}}\!Packet error rate's approx. factors \\ \!for the ARQ (HARQ) protocol\end{tabular}} & \footnotesize \!\!$\gamma_{M,j}^{(k,m)}$, $\forall k,b,j$\!\! & \rule{0pt}{12pt}\footnotesize \!\!\begin{tabular}{@{}l@{}}$17.19$\\ $(10.1)\, [\text{dB}]$\end{tabular}\!\!\!\\
\cline{5-6}
&&&& \footnotesize \!\!$g_{j}^{(k,m)}$, $\forall k,b,j$\!\! & \footnotesize \!\!$0.1$ ($0.96$)\\
\hline
\end{tabular}
\medskip
\end{center}
\end{table*}

The intrinsic and extrinsic features considered in our simulation scenario are shown in Table~\ref{tab:tab2} for a fictitious \emph{air pollution} monitoring platform that collects status packets from sixteen zones ($\mathcal{Z}_{1}$ to $\mathcal{Z}_{16}$), with different levels of importance. The goal in this platform is to maximize the average SoI provided by communicated packets over the network, following the general problem $\mathcal{P}_1$ defined in \eqref{optim1}. Without loss of generality, we consider $\mathcal{Z}_{8}$, $\mathcal{Z}_{11}$, and $\mathcal{Z}_{13}$ the most crucial zones, and the average Euclidean distance of a sensor from the centers of these three zones indicates its spatial importance. 
Furthermore, we assume that rare occurrences are monitored, where the importance of an update packet increases as its carrying pollution rate approaches $100\%$. From Table~\ref{tab:tab2}, we first find $\lambda_g^{(k,m)}=0.37$, and then, the corresponding $U_n^{(k,m)}$ for each $w_b^{(k,m)}$ is calculated. The parameters used in the simulations are summarized in Table~\ref{tab:tab3}. The values of $\gamma_{M,j}^{(k,m)}$ and $g_j^{(k,m)}$ depend on the codeword length of the $j$-th arrival according to \eqref{theta1} and \eqref{theta2}. We consider a binary phase shift keying (BPSK) modulation scheme at the PHY module. Initializing the codeword length $\ell_j$ from zero to $\ell_{\rm max}$ as the worst case, the value of $\gamma_{M,j}^{(k,m)}$ varies from $5.53$ to $5.63$ with almost constant $g_j^{(k,m)}\approx0.1$ under the ARQ protocol. We consider the average value of $\gamma_{M,j}^{(k,m)}$ for all codeword lengths, as it results in negligible difference. Applying \eqref{eq:meta}--\eqref{eq:vfb}, and Tables~\ref{tab:tab2} and \ref{tab:tab3}, meta-values are computed. % We study the $25$-th SSM and $13$-th MM in our setup. 
Unless otherwise specified, the EUT case, fixed-length filtering, ARQ protocol, and Algorithm~\ref{Alg:Alg.1} are used to obtain the results.

\subsection{Results and Discussion}
Figure~\ref{fig:res1} shows the objective function $\mathcal{J}_{\rm SoI}^{(k,m)}$ versus the input rate $\lambda_k$ and the number of selected realizations $l_k^{(m)}$ at the $k$-th sensor connected to the $m$-th MM. We observe that there is an optimal size of selected realizations for any input rate, which maximizes the SoI over the network. Otherwise stated, every random size of the selected realization set requires a specific input rate to maximize the SoI at the link level. However, the SoI saturates for large sizes of realization sets and high arrival rates. The convergence of the SoI to its maximum value with an optimal value of $l_k^{(m)} \leq n_k$ implies that the most valuable status packets can be efficiently transferred with a lower channel load. Considering the three defined forms of the utility function in \eqref{sec5:eq3}, the EUT and RUT offer the highest and lowest transferred SoI, respectively. Likewise, one can plot the SoI objective function versus the arrival rate and the size of the realization set for a system employing adaptive-length filtering, reaching to similar conclusions as in Figure~\ref{fig:res1}. 
% Optimal values of $l_k^{(m)}$ based on different input rates are shown in Table~\ref{tab:tab4} for fixed- and adaptive-length filtering under all three utility forms. Increasing the input rate results in smaller selection size to compensate a high blockage rate.% caused by heavy channel or PHY submodule load.
\begin{figure}[t!]
    \begin{minipage}{.5\textwidth}
    \centering
    \pstool[scale=0.57]{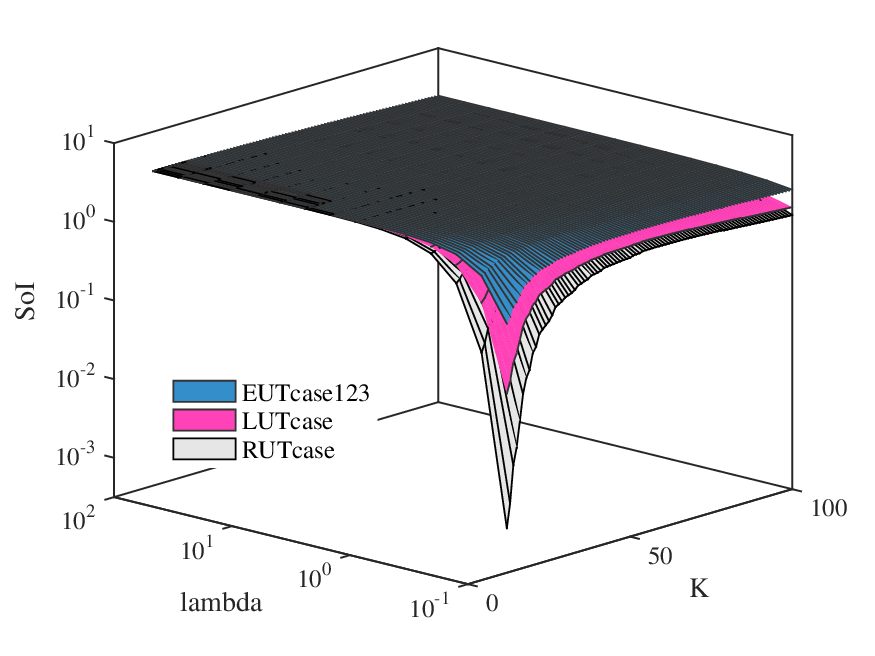}{
    \psfrag{SoI}{\hspace{-0.33cm}\footnotesize $\mathcal{J}_{\rm SoI}^{(k,m)}$}
    \psfrag{lambda}{\hspace{0.15cm}\footnotesize $\lambda_k$}
    \psfrag{K}{\hspace{-0.25cm}\footnotesize $l_k^{(m)}$}
    \psfrag{EUTcase123}{\hspace{0.03cm}\scriptsize EUT case}
    \psfrag{LUTcase}{\hspace{0.03cm}\scriptsize LUT case}
    \psfrag{RUTcase}{\hspace{0.03cm}\scriptsize RUT case}
    }
    % \vspace{-0.25cm}
    \caption{The interplay between the SoI, the arrival rate, and the admission size for the EUT, LUT, and RUT cases.}
    \label{fig:res1}
    \end{minipage}
    \begin{minipage}{.49\textwidth}
    \centering
    \pstool[scale=0.57]{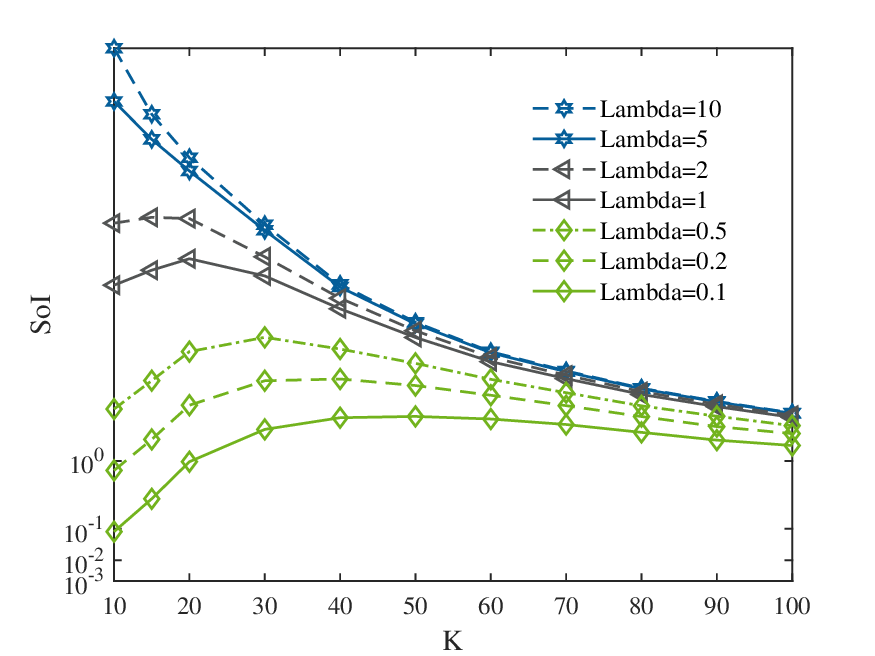}{
    \psfrag{SoI}{\hspace{-0.18cm}\footnotesize $\mathcal{J}_{\rm SoI}^{(k,m)}$}
    \psfrag{lambda}{\hspace{0.15cm}\footnotesize $\lambda_k$}
    \psfrag{K}{\hspace{-0.12cm}\footnotesize $l_k^{(m)}$}
    \psfrag{Lambda=10}{\hspace{0.03cm}\scriptsize $\lambda_k = 10$}
    \psfrag{Lambda=5}{\hspace{0.03cm}\scriptsize $\lambda_k = 5$}
    \psfrag{Lambda=2}{\hspace{0.03cm}\scriptsize $\lambda_k = 2$}
    \psfrag{Lambda=1}{\hspace{0.03cm}\scriptsize $\lambda_k = 1$}
    \psfrag{Lambda=0.5}{\hspace{0.03cm}\scriptsize $\lambda_k= 0.5$}
    \psfrag{Lambda=0.2}{\hspace{0.03cm}\scriptsize $\lambda_k = 0.2$}
    \psfrag{Lambda=0.1}{\hspace{0.03cm}\scriptsize $\lambda_k = 0.1$}
    }
    % \vspace{-0.25cm}
    \caption{SoI versus the size of the admitted realizations for the EUT case and different arrival rates.}
    \label{fig:res2BB}
    \end{minipage}
\end{figure}
Figure~\ref{fig:res2BB}, a two-dimensional depiction of Figure~\ref{fig:res1} for the EUT case, shows the objective function $\mathcal{J}_{\rm SoI}^{(k,m)}$ versus the selection size $l_k^{(m)}$ for different values of $\lambda_k^{(m)}$. Expectedly, higher arrival rates lead to smaller sets of selected realizations, as fewer packets are selected to avoid a high blockage rate, whereas, for low arrival rates, a large number of packets can pass the filter to harness high lateness.

\begin{figure}
\centering
    \pstool[scale=0.57]{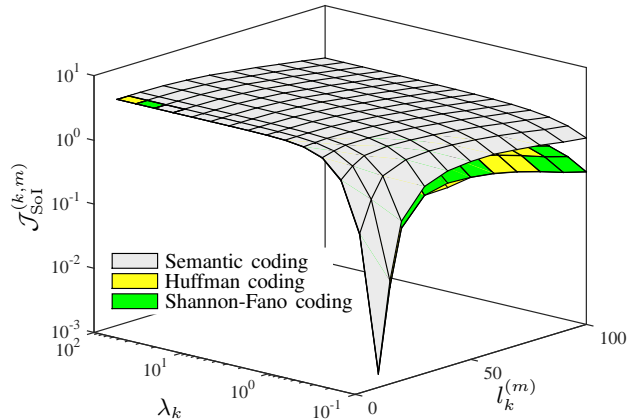}{
    \psfrag{SoI}{\hspace{-0.37cm}\footnotesize $\mathcal{J}_{\rm SoI}^{(k,m)}$}
    \psfrag{lambda}{\hspace{0.15cm}\footnotesize $\lambda_k$}
    \psfrag{K}{\hspace{-0.25cm}\footnotesize $l_k^{(m)}$}
    \psfrag{Semantic-Coding}{\hspace{0.03cm}\scriptsize Semantic coding}
    \psfrag{Huffman-Coding}{\hspace{0.03cm}\scriptsize Huffman coding}
    \psfrag{Shannon-Fano-Coding}{\hspace{0.03cm}\scriptsize Shannon-Fano coding}
    }
    % \vspace{-0.25cm}
    \caption{The comparison between the SoI performances of the proposed semantic, Shannon-Fano, and Huffman source coding schemes for the RUT case.}
    \label{fig:res1R}
\end{figure}
Besides, the performance of the proposed semantic coding is compared with those of \emph{Shannon-Fano} and \emph{Huffman} coding in Figure~\ref{fig:res1R} for the RUT case. It is shown that the optimal real-valued codeword lengths derived herein result in higher SoI compared to the other coding schemes. Moreover, as the size of the realization or selection set increases, the performance gap increases. The reason is that the value of observations and the timeliness of the system are not taken into account in either Shannon-Fano or Huffman coding.

Initializing the arrival rate from $10^{-1}$ to $10^2$ and finding optimal values for $l_k^{(m)}$, we assess the performance of the proposed semantics-aware filtering techniques. For that, Figure~\ref{fig:res2} depicts the sum SoI objective function $\mathcal{J}_{{\rm SoI}} \coloneqq \sum_{m=1}^{M} \sum_{k \in \mathcal{K}(m)} \mathcal{J}_{{\rm SoI}}^{(k,m)}$ versus input rate $\lambda_1=...=\lambda_K=\lambda_k$ for the fixed-length, adaptive-length, and an asymptotic variant of adaptive-length with $l_k^{(k)}=n_k$, $\forall k$. We see that the asymptotic adaptive-length method outperforms the fixed-length one for low arrival rates. However, the performance of the asymptotic adaptive-length filtering decreases with increasing the arrival rate. The performance of the fixed-length filtering gradually improves by increasing the input rate since the size of admitted packets shrinks to lower the blockage rate. Thanks to its flexibility, adaptive-length filtering outperforms its asymptotic variant and fixed-length filtering for any arrival rate. %Specifically, it starts with the same performance as the asymptotic adaptive-length filtering for low input rates; finally, it converges to that of the fixed-length filtering as the input rate increases. 
\begin{figure}[t!]
    \centering
    \pstool[scale=0.57]{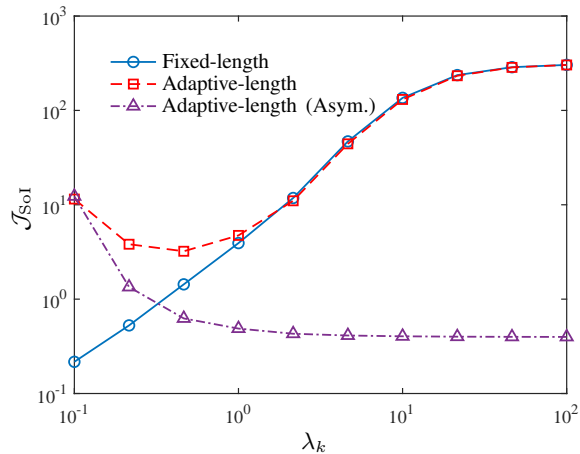}{
    \psfrag{SoI}{\hspace{-0.17cm}\footnotesize $\mathcal{J}_{\rm SoI}$}
    \psfrag{lambda}{\hspace{0.15cm}\footnotesize $\lambda_k$}
    \psfrag{K}{\hspace{-0.25cm}\footnotesize $l_k^{(m)}$}
    \psfrag{Selectivefiltering}{\hspace{0.03cm}\scriptsize Fixed-length}
    \psfrag{Adaptivefiltering}{\hspace{0.03cm}\scriptsize Adaptive-length (Asym.)}
    \psfrag{Hybridfiltering}{\hspace{0.03cm}\scriptsize Adaptive-length}
    }
    % \vspace{-0.25cm}
    \caption{SoI performance of semantics-aware filtering versus the arrival rate.}
    \label{fig:res2}
\end{figure}

The normalized SoI of the transmitted and monitored update packets are plotted in Figure~\ref{fig:res3}, simulated over an interval of $[0, 1000]\,[\text{sec.}]$ for the rates of $\lambda_k=0.5$ and $5$, which validates the analytical results used for plotting Figure~\ref{fig:res2}.  
\begin{figure}[t!]
    \begin{subfigure}[]{0.49\textwidth}
    \centering
    \pstool[scale=0.57]{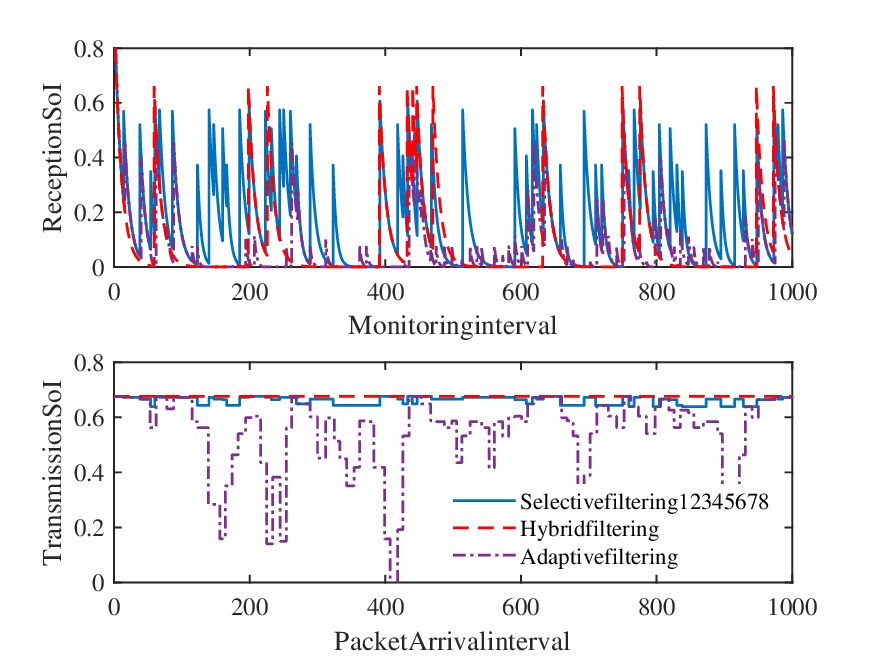}{
    \psfrag{Monitoringinterval}{\hspace{-0.35cm}\footnotesize $\text{Monitoring interval}\,[\text{sec}]$}
    \psfrag{PacketArrivalinterval}{\hspace{-0.03cm}\footnotesize $\text{Arrival interval}\,[\text{sec}]$}
    \psfrag{ReceptionSoI}{\hspace{-0.7cm}\footnotesize Monitored norm. SoI}
    \psfrag{TransmissionSoI}{\hspace{-0.71cm}\footnotesize Transmitted norm. SoI}
    \psfrag{Selectivefiltering12345678}{\hspace{0.03cm}\scriptsize Fixed-length}
    \psfrag{Adaptivefiltering}{\hspace{0.03cm}\scriptsize Adaptive-length (Asym.)}
    \psfrag{Hybridfiltering}{\hspace{0.03cm}\scriptsize Adaptive-length}
    }
    \caption{}
    \end{subfigure}
    \hfill
    \begin{subfigure}[]{0.5\textwidth}
    \centering
    \pstool[scale=0.57]{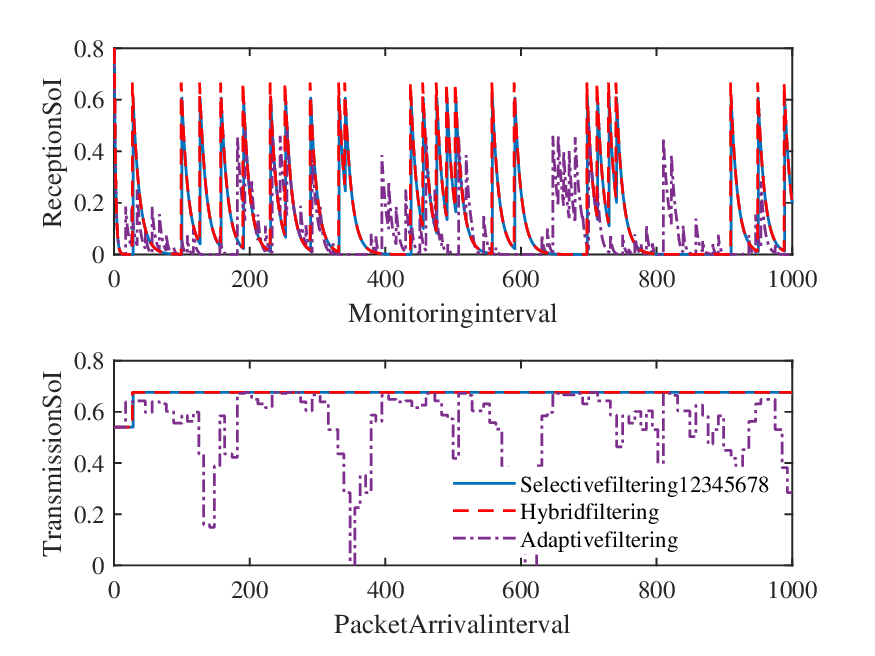}{
    \psfrag{Monitoringinterval}{\hspace{-0.35cm}\footnotesize $\text{Monitoring interval}\,[\text{sec}]$}
    \psfrag{PacketArrivalinterval}{\hspace{-0.03cm}\footnotesize $\text{Arrival interval}\,[\text{sec}]$}
    \psfrag{ReceptionSoI}{\hspace{-0.68cm}\footnotesize Monitored norm. SoI}
    \psfrag{TransmissionSoI}{\hspace{-0.71cm}\footnotesize Transmitted norm. SoI}
    \psfrag{Selectivefiltering12345678}{\hspace{0.03cm}\scriptsize Fixed-length}
    \psfrag{Adaptivefiltering}{\hspace{0.03cm}\scriptsize Adaptive-length (Asym.)}
    \psfrag{Hybridfiltering}{\hspace{0.03cm}\scriptsize Adaptive-length}
    }
    \caption{}
    \end{subfigure}
    % \vspace{-0.25cm}
    \caption{Comparison of the proposed semantics-aware filtering from the transmitted and monitored SoI perspectives for (a) $\lambda_k = 0.5$ and (b) $\lambda_k=5$.}
    \label{fig:res3}
\end{figure}

Using the simulations in Figure~\ref{fig:res3}, Figure~\ref{fig:res1table} depicts the average percentages (rates) of filtered, blocked, and transmitted packets to the $m$-th MM, out of all packets arrived at the $k$-th SSM within the interval of $[0, 1000]\,[\text{sec.}]$, averaged over $500$ iterations. The adaptive-length method has higher filtering rates, i.e., $70.47\%$ ($51.13\%$), and provides lower blockage rates, i.e., $19.93\%$ ($46.45\%$), compared to those of the fixed-length filtering, i.e., the filtering rates of $65.68\%$ ($45.66\%$) and the blockage rates of $24.78\%$ ($52.17\%$), for $\lambda_k=0.5$ ($5$), thanks to its flexibility and semantics-aware decisions. The advantages of adaptive-length filtering come at the cost of requiring slightly higher channel capacity, i.e., $0.06\%$ and $0.25\%$ for $\lambda_k=0.5$ and $5$, respectively, compared to a fixed-length filter. It is worth mentioning that the drops of the transmitted SoI for the asymptotic adaptive-length filtering in Figure~\ref{fig:res3} occur due to elevated blockage rate and number of consecutive blockages. 
\begin{figure}[t!]
\begin{minipage}{.5\textwidth}
    \centering
    \pstool[scale=0.57]{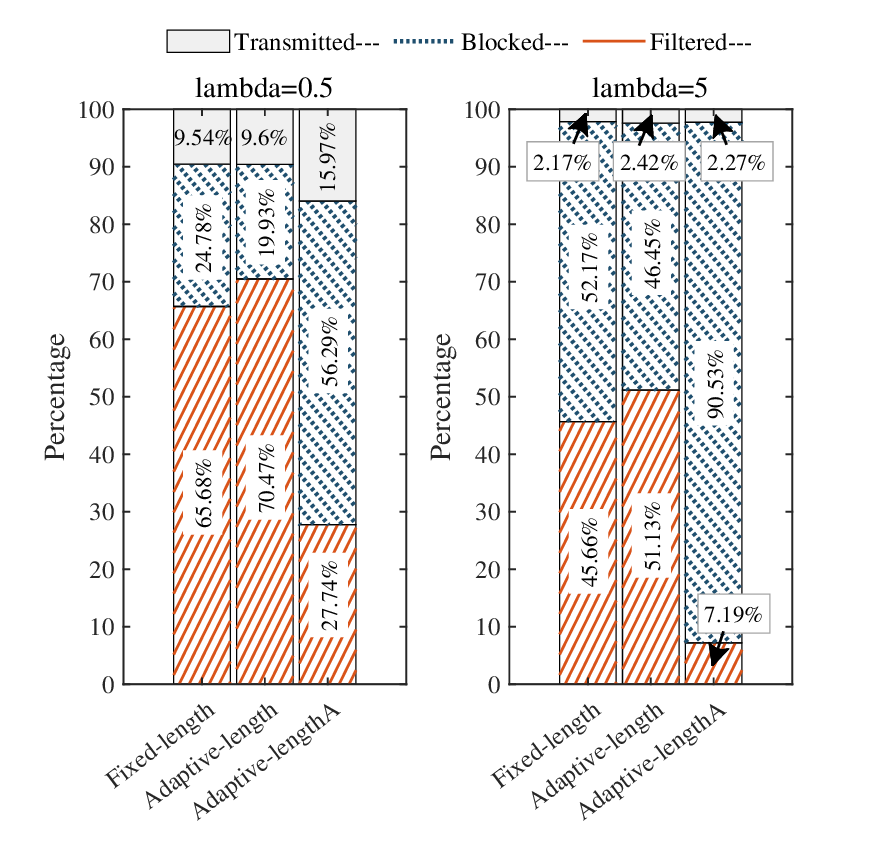}{
    \psfrag{Percentage}{\hspace{-2.07cm}\footnotesize Average percentages of processed packets}
    \psfrag{Transmitted---}{\hspace{0.03cm}\scriptsize Transmitted}
    \psfrag{Blocked---}{\hspace{0.03cm}\scriptsize Blocked}
    \psfrag{Filtered---}{\hspace{0.03cm}\scriptsize Filtered}
    \psfrag{Fixed-length}{\hspace{-0.1cm}\scriptsize Fixed-length}
    \psfrag{Adaptive-length}{\hspace{-0.1cm}\scriptsize Adaptive-length}
    \psfrag{Adaptive-lengthA}{\hspace{-0.92cm}\scriptsize Adaptive-length (Asym.)}
    \psfrag{lambda=0.5}{\hspace{0.09cm}\scriptsize $\lambda_k = 0.5$}
    \psfrag{lambda=5}{\hspace{0.09cm}\scriptsize $\lambda_k = 5$}
    }
    \vspace{0.01cm}
    \caption{Average percentage of filtered, blocked, and transmitted packets arrived within the interval of $[0, 1000]\,[\text{sec.}]$ under different filtering techniques.}
    \label{fig:res1table}
\end{minipage}
\begin{minipage}{.49\textwidth}
    \centering
    \pstool[scale=0.57]{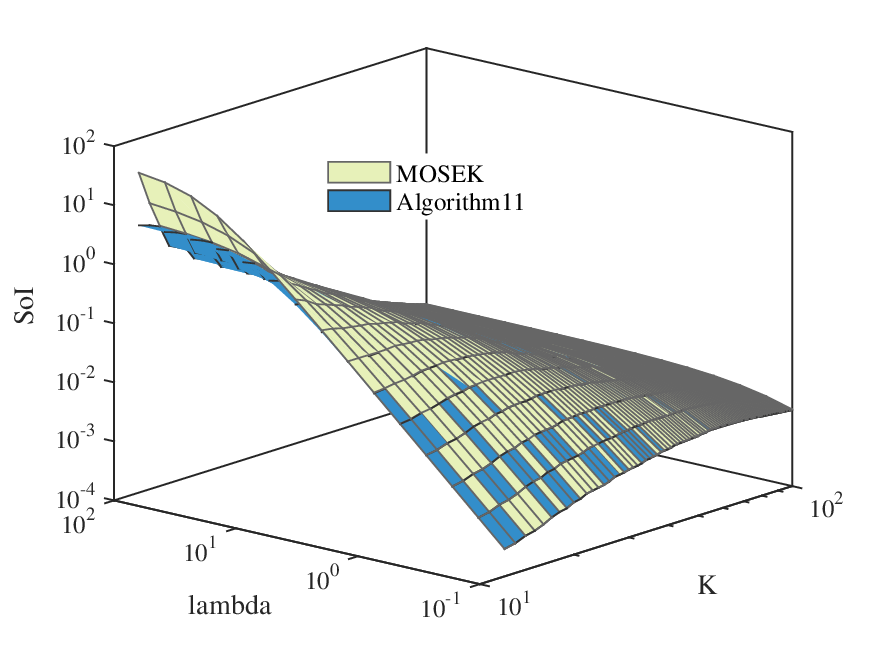}{
    \psfrag{SoI}{\hspace{-0.33cm}\footnotesize $\mathcal{J}_{\rm SoI}^{(k,m)}$}
    \psfrag{lambda}{\hspace{0.15cm}\footnotesize $\lambda_k$}
    \psfrag{K}{\hspace{-0.25cm}\footnotesize $l_k^{(m)}$}
    \psfrag{MOSEK}{\hspace{0.03cm}\scriptsize MOSEK solver}
    \psfrag{Algorithm11}{\hspace{0.03cm}\scriptsize Algorithm~\ref{Alg:Alg.1}}
    \psfrag{RUTcase}{\hspace{0.03cm}\scriptsize RUT case}
    }
    % \vspace{-0.25cm}
    \caption{SoI versus the arrival rate and the admission size for Algorithm~\ref{Alg:Alg.1} and the MOSEK solver.}
    \label{fig:res4}
\end{minipage}
\end{figure}

Figure~\ref{fig:res4} depicts the SoI versus $\lambda_k$ and $l_k^{(m)}$ for two scenarios, in order to evaluate the accuracy of Algorithm~\ref{Alg:Alg.1}. In the first scenario, the MOSEK solver is used to find the optimal codeword lengths; in the second scenario, Algorithm~\ref{Alg:Alg.1} is applied to determine the optimal codeword lengths from \eqref{sec5:eq9}, \eqref{sec5:eq14}, or \eqref{sec5:eq17}. We observe that the proposed algorithm and the solver perform identically for low to moderate arrival rates, whereas for higher rates, the MOSEK solver slightly outperforms Algorithm~\ref{Alg:Alg.1}, resulting in a difference in SoI as low as $0.32\%$ on average. Nevertheless, the algorithm has the advantage of tractability, providing design insights and information about the effect of different system parameters on the codeword lengths.
% \begin{figure}[t!]
%     \centering
%     \pstool[scale=0.57]{Media/SoI_vs_K_Lambda_MOSEKandalgorithmB.eps}{
%     \psfrag{SoI}{\hspace{-0.33cm}\footnotesize $\mathcal{J}_{\rm SoI}^{(k,m)}$}
%     \psfrag{lambda}{\hspace{0.15cm}\footnotesize $\lambda_k$}
%     \psfrag{K}{\hspace{-0.25cm}\footnotesize $l_k^{(m)}$}
%     \psfrag{MOSEK}{\hspace{0.03cm}\scriptsize MOSEK solver}
%     \psfrag{Algorithm11}{\hspace{0.03cm}\scriptsize Algorithm~\ref{Alg:Alg.1}}
%     \psfrag{RUTcase}{\hspace{0.03cm}\scriptsize RUT case}
%     }
%     % \vspace{-0.25cm}
%     \caption{SoI versus the arrival rate and the admission size for Algorithm~\ref{Alg:Alg.1} and the MOSEK solver.}
%     \label{fig:res4}
% \end{figure}

In Figure~\ref{fig:res5}\,\textcolor{red}{(a)} and Figure~\ref{fig:res5}\,\textcolor{red}{(b)}, we compare the SoI when ARQ and HARQ protocols are used respectively under two different truncation models. The first model has a fixed (predefined) maximum number of transmission rounds, i.e., $r_{{\rm max}, j}^{(k,m)} = r_{\rm max}$, while in the second model, the importance of the arrival affects the maximum possible number of its transmissions according to \eqref{eq:rmaxeq}. It is shown that both models provide almost the same SoI for small admission sizes. However, the semantics-aware truncated protocol exhibits higher performance for moderate to large admission sizes. Indeed, when the arrival traffic increases, reserving the channel for more important packets results in higher gain compared to giving an equal chance to all blocked packets to be retransmitted for a fixed number of rounds, regardless of their usefulness. 
\begin{figure}[t!]
    \begin{subfigure}[]{0.49\textwidth}
    \centering
    \pstool[scale=0.57]{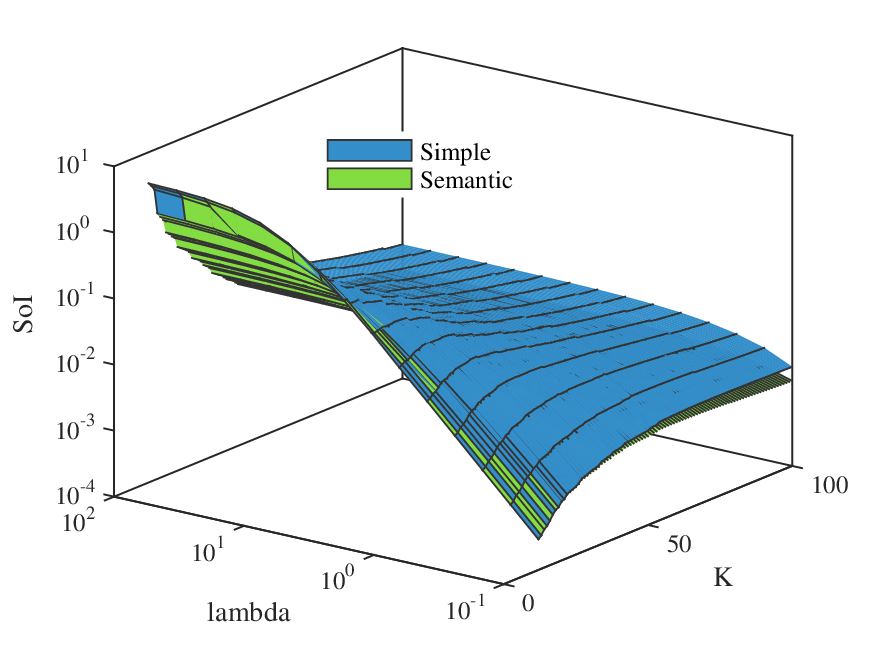}{
    \psfrag{SoI}{\hspace{-0.33cm}\footnotesize $\mathcal{J}_{\rm SoI}^{(k,m)}$}
    \psfrag{lambda}{\hspace{0.15cm}\footnotesize $\lambda_k$}
    \psfrag{K}{\hspace{-0.25cm}\footnotesize $l_k^{(m)}$}
    \psfrag{Simple}{\hspace{0.03cm}\scriptsize Semantics-aware}
    \psfrag{Semantic}{\hspace{0.03cm}\scriptsize Fixed $r_{\rm max}$}
    }
    \caption{}
    \end{subfigure}
    \hfil
    \begin{subfigure}[]{0.5\textwidth}
    \centering
    \pstool[scale=0.57]{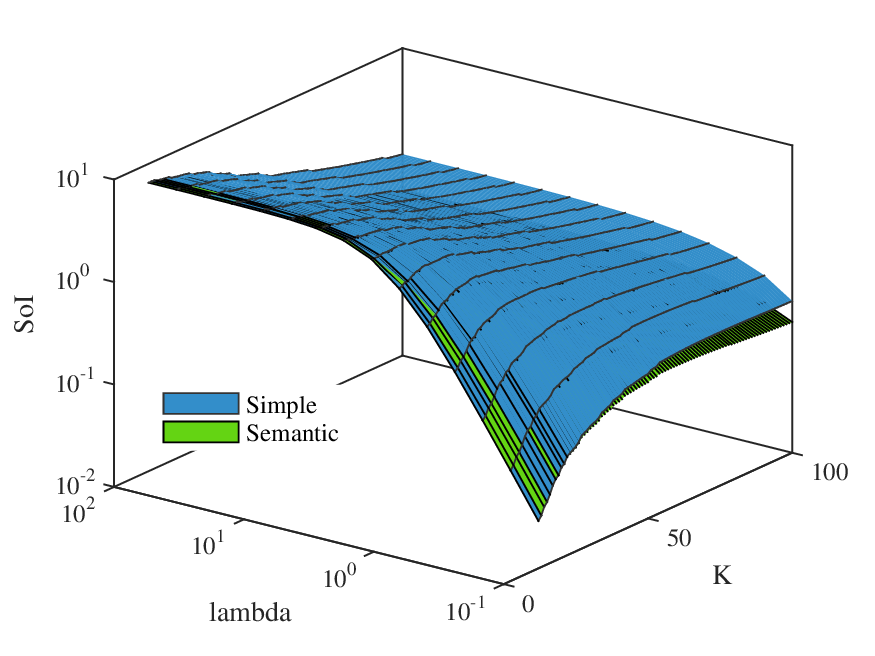}{
    \psfrag{SoI}{\hspace{-0.33cm}\footnotesize $\mathcal{J}_{\rm SoI}^{(k,m)}$}
    \psfrag{lambda}{\hspace{0.15cm}\footnotesize $\lambda_k$}
    \psfrag{K}{\hspace{-0.25cm}\footnotesize $l_k^{(m)}$}
    \psfrag{Simple}{\hspace{0.03cm}\scriptsize Semantics-aware}
    \psfrag{Semantic}{\hspace{0.03cm}\scriptsize Fixed $r_{\rm max}$}
    }
    \caption{}
    \end{subfigure}
    % \vspace{-0.25cm}
    \caption{SoI performance of (a) ARQ and (b) HARQ error control protocols truncated with fixed and semantics-aware transmission limits.}
    \label{fig:res5}
\end{figure}

\section{Conclusion}\label{Section7}
We studied the problem of timely source coding in a distributed monitoring network with heterogeneous goals, where observations at each information source are filtered and transmitted to monitoring entities, depending on their importance (semantics) for achieving each monitor's goal. 
We analytically derived the real-valued optimal codeword lengths that maximize a weighted sum of semantics-aware utility functions. Specifically, our analytical and numerical results show an optimal number of realizations to send for each link, which depends, among others, on the source statistics, the packet arrival rate, the semantic value of each arrival, and the weight assigned to each monitor’s goal. We proposed two different semantic filtering methods and showed that adaptive-length filtering may provide higher SoI with a lower channel blockage rate at the cost of a minor increase in the transmission load. We also extended the conventional truncated error control protocols to a semantics-aware one that specifies the maximum number of retransmissions based on the importance of a packet.
The main takeaway is that properly optimized semantic filtering and source coding can significantly reduce the number of status updates required for the timely delivery of semantic information to monitors with heterogeneous goals. 

\appendices
\section{Proof of Lemma~\ref{lemm1}}\label{lemm:app1}
At a layer of the $k$-th SSM serving the $m$-th MM, the timeliness utility of the $(j\!+\!d)$-th packet crosses the $j$-th one if 
\begin{align} \label{eq1:lemma1:t}
    \rho_j^{(k,m)} \left(t_c^{j,d}\!-\!t_{j}^{(k,m)}\right) = \rho_{j+d}^{(k,m)}  \left(t_c^{j,d}\!-\!t_{j+d}^{(k,m)}\right)\!,
\end{align}
as in Figure~\ref{fig:fig3}, for every type of the defined forms in \eqref{sec5:eq3}, where $t_c^{j,d}$ denotes the cross time of order $d$ for the $j$-th arrival. From \eqref{eq1:lemma1:t}, we can write
\begin{equation}
    \rho_j^{(k,m)} \left(t_c^{j,d}\!-\!t_{j}^{(k,m)}\right) = \rho_{j+d}^{(k,m)} \!\left(t_c^{j,d}\!-\!t_{j}^{(k,m)}\!-\!\sum_{k=j+1}^{j+d}Y_k^{(k,m)}\!\right)\!
\end{equation}
which results in
\begin{align} \label{eq3:lemma1:t}
   t_c^{j,d} = t_j^{(k,m)} + \dfrac{\rho_{j+d}^{(k,m)}}{\rho_{j+d}^{(k,m)} \!-\! \rho_{j}^{(k,m)}} \sum_{k=j+1}^{j+d}Y_k^{(k,m)}.
\end{align}
Since $f_k(\cdot)$ is convex and non-increasing, if the $(j\!+\!d)$-th arrivals' utility plot crosses that of the $j$-th arrival, the $j$-th packet provides higher or equal SoI compared to the $(j\!+\!d)$-th packet, beyond that cross point. Therefore, the semantic filter discards the $(j\!+\!d)$-th status packet in the presence of the $j$-th one if $t_c^{j,d}$ is before the end of $S_{j+d}^{(k,m)}$. From \eqref{eq3:lemma1:t}, we reach
\begin{align} \label{eq4:lemma1:t} 
t_j^{(k,m)} + \dfrac{\rho_{j+d}^{(k,m)}}{\rho_{j+d}^{(k,m)} \!-\! \rho_{j}^{(k,m)}} \sum_{k=j+1}^{j+d}Y_k^{(k,m)} < t_j^{(k,m)} + \sum_{k=j+1}^{j+d}Y_k^{(k,m)} + S_{j+d}^{(k,m)},
\end{align}
then we obtain
\begin{align} \label{eq5:lemma1:t} 
    \dfrac{\rho_{j}^{(k,m)}}{\rho_{j+d}^{(k,m)} \!-\! \rho_{j}^{(k,m)}} < \dfrac{S_{j+d}^{(k,m)}}{\sum_{k=j+1}^{j+d}(S_k^{(k,m)}\!+\!W_{0,k}^{(k,m)})} \leq \dfrac{\ell_{\rm max}}{d \ell_{\rm max}+\sum_{k=j+1}^{j+d}W_{0,k}^{(k,m)}} \stackrel{(\ast)}{=} \dfrac{\ell_{\rm max}}{d \ell_{\rm max}+\widehat{W}_d^{(k,m)}}
\end{align}
where $\ell_{\rm max}$ denotes the upper bound for the size of a codeword length, and $W_{0,k}^{(k,m)}$ is the waiting time of the $k$-th arrival prior to filtering, which follows an exponential distribution with rate $\lambda_kq_{l_k^{(m)}}$. For $(\ast)$, we define $\widehat{W}_d^{(k,m)} \coloneqq \sum_{k=j+1}^{j+d}W_{0,k}^{(k,m)}$. Since $d$ is integer, $\widehat{W}_d^{(k,m)}$ becomes an Erlang distributed r.v. of order $d$ and rate $\lambda_kq_{l_k^{(m)}}$. With regard to the definition of $\psi_k^{(m)}$ in \eqref{sec4:eq1}, and after some manipulations, we obtain \eqref{sec4:eq2}.

\section{Proof of Lemma~\ref{lemm2}}\label{lemm:app2}
With regard to the form of $f_k(\Delta_k^{(m)}(t))$ for the EUT case in \eqref{sec5:eq3} and Figure~\ref{fig:fig4}, we compute polygons $Q_{j}^{(k,m)}$, $\forall j\geq2$, as 
\begin{align}\label{sec5:eq6}
Q_{j}^{(k,m)} &= \int_{t_{j-1}^{(k,m)}}^{t_{j-1}^{(k,m)}+S_{j}^{(k,m)}} \!\!\left(e^{-\rho_{j-1}^{(k,m)} (t-t_{j-2}^{(k,m)})}  \!+\! \beta_k^{(m)}\right)\!{\rm d}t +\int_{t_{j-1}^{(k,m)} + S_j^{(k,m)}}^{t_{j}^{(k,m)}} \!\!\left(e^{-\rho_j^{(k,m)}(t-t_{j-1}^{(k,m)})} \!+\! \beta_k^{(m)}\right)\!{\rm d}t \nonumber \\
&\approx - \rho_{j-1}^{(k,m)}Y_{j-1}^{(k,m)}S_{j}^{(k,m)} + \dfrac{1}{2} (\rho_{j}^{(k,m)} \!-\! \rho_{j-1}^{(k,m)})(S_{j}^{(k,m)})^2 - \dfrac{1}{2}\rho^{(k,m)}_{j}(Y_{j}^{(k,m)})^2 + (1\!+\!\beta_k^{(m)})Y_{j}^{(k,m)}
\end{align}
where the approximation comes from second-order Taylor expansion for the exponential function. The expected value is then given by
\begin{align}\label{sec5:eq7}
    \mathbb{E}[Q^{(k,m)}] &\approx 
   - \mathbb{E}\!\left[\rho^{(k,m)}Y^{(k,m)}\right]\mathbb{E}[S^{(k,m)}] + \dfrac{1}{2}
   \mathbb{E}\!\left[\rho^{(k,m)}(S^{(k,m)})^2\right]  \nonumber \\
   &~~~- \dfrac{1}{2}
   \mathbb{E}[\rho^{(k,m)}]\mathbb{E}\!\left[(S^{(k,m)})^2\right]  - \dfrac{1}{2}\mathbb{E}\!\left[\rho^{(k,m)}(Y^{(k,m)})^2\right] + (1\!+\!\beta_k^{(m)})\mathbb{E}[Y^{(k,m)}].
\end{align}
Note that large enough $\beta_k^{(m)}$ guarantees the positiveness of $\mathbb{E}[Q^{(k,m)}]$ in \eqref{sec5:eq7} after applying the expansion. By importing $Y_j^{(k,m)} = W_j^{(k,m)} + S_j^{(k,m)}$ and $\mathbb{E}[(S^{(k,m)})^m]$ from \eqref{sec2:eq2} into \eqref{sec5:eq7}, and since $\rho_j^{(k,m)}$ and $W_j^{(k,m)}$ are independent, we obtain \eqref{sec5:eq8}.

\section{Proof of Proposition~\ref{prop1}}\label{app1}
Using the parameter $\chi_k^{(m)}$ given in \eqref{sec5:eq11}, we rewrite \eqref{sec5:eq8} as follows.
\begin{align}\label{app1:eq1} 
    \mathbb{E}[Q^{(k,m)}] &\approx 
    - \chi_k^{(m)} \bar{\rho}_k^{(m)} \!\left(\mathbb{E}\!\left[\varphi^{(k,m)}L^{(k,m)}\right]\right)^{\!2} - \dfrac{1}{2}
   \bar{\rho}_k^{(m)}\mathbb{E}\!\left[\varphi^{(k,m)}(L^{(k,m)})^2\right] - \chi_k^{(m)} \bar{\rho}_k^{(m)} \gamma_k^{(m)} \mathbb{E}\!\left[\varphi^{(k,m)}L^{(k,m)}\right] \nonumber \\
    &~~~+ \Big( 1\!+\!\beta_k^{(m)} \!-\! \bar{\rho}_k^{(m)}\gamma_k^{(m)}\Big)\mathbb{E}\!\left[\varphi^{(k,m)}L^{(k,m)}\right] - \bar{\rho}_k^{(m)}(\gamma_k^{(m)})^2 + (1\!+\!\beta_k^{(m)})\gamma_k^{(m)}.
\end{align}

By \eqref{app1:eq1}, we define the Lagrange function $\mathcal{L}(\ell_i^{(k,m)};\mu_k^{(m)})$ for $\mathcal{P}_2$.
% as
% \begin{align}\label{app1:eq2} \nonumber
%     \mathcal{L}(\ell_i^{(k,m)};\mu_k^{(m)}) &= - \sum_{m=1}^{M} w_m \sum_{k \in \mathcal{K}(m)} p_k^{(m)}\eta_k^{(m)} \times\!\Bigg[\chi_k^{(m)} \bar{\rho}_k^{(m)} \Big(\!\sum_{i \in \mathcal{I}_{l_k^{(m)}}} \! p_i^{(k,m)} \varphi_i^{(k,m)} \ell_i^{(k,m)}\Big)^{\!2}\\ \nonumber
%     &~~~- \dfrac{1}{2}
%    \bar{\rho}_k^{(m)}\!\sum_{i \in \mathcal{I}_{l_k^{(m)}}}  \! p_i^{(k,m)} \varphi_i^{(k,m)} (\ell_i^{(k,m)})^2 - \chi_k^{(m)} \bar{\rho}_k^{(m)} \gamma_k^{(m)} \! \sum_{i \in \mathcal{I}_{l_k^{(m)}}} \! p_i^{(k,m)} \varphi_i^{(k,m)} \ell_i^{(k,m)}\\ \nonumber
%     &~~~+ \Big( 1\!+\!\beta_k^{(m)} \!-\! \bar{\rho}_k^{(m)}\gamma_k^{(m)}\Big) \! \sum_{i \in \mathcal{I}_{l_k^{(m)}}} \! p_i^{(k,m)} \varphi_i^{(k,m)} \ell_i^{(k,m)}\\
%    &~~~- \bar{\rho}_k^{(m)}(\gamma_k^{(m)})^2 + (1\!+\!\beta_k^{(m)})\gamma_k^{(m)}\Bigg] - \mu_k^{(m)} \bigg(\!\sum_{i\in \mathcal{I}_{l_k^{(m)}}} 2^{-\ell_i^{(k,m)}}\!-\!1\bigg)
% \end{align}
% where $\mu_k^{(m)}\geq0$ is the Lagrange multiplier. 
Next, we can write Karush-Kuhn-Tucker (KKT) conditions for $ i\in\mathcal{I}_{l_k^{(m)}}$, as what follows.
\begin{align}\label{app1:eq3}
    &\frac{\partial \mathcal{L}(\ell_i^{(k,m)};\mu_k^{(m)})}{\partial \ell_i^{(k,m)}} = - 2\chi_k^{(m)} \bar{\rho}_k^{(m)}p_i^{(k,m)} \varphi_i^{(k,m)} \!\!\!\sum_{i \in \mathcal{I}_{l_k^{(m)}}} \!\! p_i^{(k,m)} \varphi_i^{(k,m)} \ell_i^{(k,m)} 
    - \bar{\rho}_k^{(m)} p_i^{(k,m)} \varphi_i^{(k,m)} \ell_i^{(k,m)} \nonumber \\
    &- \chi_k^{(m)} \bar{\rho}_k^{(m)} \gamma_k^{(m)} \! p_i^{(k,m)} \varphi_i^{(k,m)} + \Big( 1\!+\!\beta_k^{(m)} \!-\! \bar{\rho}_k^{(m)}\gamma_k^{(m)}\Big) p_i^{(k,m)} \varphi_i^{(k,m)} + \mu_k^{(m)} \ln(2)2^{-\ell_i^{(k,m)}}  = 0
\end{align}
after combining $w_m p_k^{(m)}\eta_k^{(m)}$ with $\mu_k^{(m)}$, where $\mu_k^{(m)}\geq0$ is the Lagrange multiplier. The complementary slackness condition is 
\begin{align}\label{app1:eq4}
    \mu_k^{(m)} \bigg(\!\sum_{i\in \mathcal{I}_{l_k^{(m)}}} 2^{-\ell_i^{(k,m)}}\!-\!1\bigg)\!=0.
\end{align}
There exist two conditions, one of which meets \eqref{app1:eq4}, as below.
\begin{itemize}
    \item[(i)] $\mu_k^{(m)}=0$, hence $\sum_{i\in \mathcal{I}_{l_k^{(m)}}} 2^{-\ell_i^{(k,m)}} < 1$; or
    \item[(ii)] $\mu_k^{(m)} \neq 0$, hence $\sum_{i\in \mathcal{I}_{l_k^{(m)}}} 2^{-\ell_i^{(k,m)}}=1$.
\end{itemize}
Under condition (i), the left side of \eqref{app1:eq4} may potentially lead to non-positive codeword lengths, hence $\mathbb{E}\big[\varphi^{(k,m)}L^{(k,m)}\big]\leq0$, for a given $\beta_k^{(m)}$. Since negative codeword lengths are not meaningful, by contradiction, condition (ii) must satisfy \eqref{app1:eq4} where parameter $\mu_k^{(m)}$ offers a degree of flexibility. 
After some algebraic manipulations, we reach
\begin{align}\label{app1:eq5}
    \dfrac{\mu_k^{(m)} (\ln(2))^2}{ \bar{\rho}_k^{(m)} p_i^{(k,m)} \varphi_i^{(k,m)}}2^{- \ell_i^{(k,m)}} \operatorname{exp}\!\left(\!\dfrac{\mu_k^{(m)} (\ln(2))^2}{\bar{\rho}_k^{(m)} p_i^{(k,m)} \varphi_i^{(k,m)}} 2^{-\ell_i^{(k,m)}}\!\right)  = \dfrac{\mu_k^{(m)} (\ln(2))^2}{\bar{\rho}_k^{(m)} p_i^{(k,m)} \varphi_i^{(k,m)}} 2^{{\xi_k^{(m)}}}
\end{align}
where $\xi_k^{(m)}$ is given in \eqref{sec5:eq10}.
Herein, the form of \eqref{app1:eq5} is equal to $x \operatorname{exp}(x) = y$ for which the solution is $x = W_m(y)$, where $m=0$ for $y\geq0$. 

\section{Proof of Lemma~\ref{lemm3}}\label{lemm:app3}
Considering the form of $f_k(\Delta_k^{(m)}(t))$ for the LUT case in \eqref{sec5:eq3}, polygons $Q_{j}^{(k,m)}$, $\forall j\geq2$, and using Taylor expansion, yields
\begin{align}\label{sec5:eq12}
&Q_{j}^{(k,m)} = \int_{t_{j-1}^{(k,m)}}^{t_{j-1}^{(k,m)}+S_{j}^{(k,m)}} \!\!\left(\ln\!\big(\!-\!\rho_{j-1}^{(k,m)}(t-t_{j-2}^{(k,m)})\big) \!+\! \beta_k^{(m)}\right)\!{\rm d}t \nonumber \\ &~~~~~~~~ +\int_{t_{j-1}^{(k,m)} + S_j^{(k,m)}}^{t_{j}^{(k,m)}} \!\!\left(\ln\!\big(\!-\!\rho_{j}^{(k,m)}(t-t_{j-1}^{(k,m)})\big) \!+\! \beta_k^{(m)}\right)\!{\rm d}t \nonumber\\
&\approx - 2\rho_{j-1}^{(k,m)}Y_{j-1}^{(k,m)}S_{j}^{(k,m)} + (\rho_{j}^{(k,m)} \!-\! \rho_{j-1}^{(k,m)})(S_{j}^{(k,m)})^2  - \rho^{(k,m)}_{j}(Y_{j}^{(k,m)})^2 + (\beta_k^{(m)} \!-\!1)Y_{j}^{(k,m)},
\end{align}
whose expectation is given by
\begin{align}\label{sec5:eq12b}
    \mathbb{E}[Q^{(k,m)}] &\approx 
   - 2\mathbb{E}\!\left[\rho^{(k,m)}Y^{(k,m)}\right]\mathbb{E}[S^{(k,m)}] + 
   \mathbb{E}\!\left[\rho^{(k,m)}(S^{(k,m)})^2\right] \nonumber \\
   &~~~- 
   \mathbb{E}[\rho^{(k,m)}]\mathbb{E}\!\left[(S^{(k,m)})^2\right]  - \mathbb{E}\!\left[\rho^{(k,m)}(Y^{(k,m)})^2\right]+ (\beta_k^{(m)} \!-\! 1)\mathbb{E}[Y^{(k,m)}].
\end{align}
Using \eqref{sec5:eq12b}, and after some mathematical manipulations, we obtain \eqref{sec5:eq13}.

\bibliographystyle{IEEEtran}
\bibliography{References.bib}

\balance

\end{document}